\documentclass{llncs}
\usepackage{enumitem}
\usepackage{amsfonts}
\usepackage{amsmath}
\usepackage{mathrsfs}
\usepackage{eufrak}
\usepackage{times}
\usepackage{bbding}
\usepackage{subfigure}
\usepackage{graphicx}
\usepackage{tabularx}
\usepackage[table]{xcolor}

\makeatletter
\def\namedlabel#1#2{\begingroup
   \def\@currentlabel{#2}%
   \label{#1}\endgroup
}
\makeatother

\newcommand{\ftstruct}{\tstruct}
\newenvironment{algo}[2]%
{\par\vspace{1em}\noindent\textbf{Algorithm~#1}~(\emph{#2}).~~}
{}

\usepackage{xspace}

\usepackage{color}

\newcommand{\ltlG}{\mathsf{G}}

\newcommand{\ltlU}{\mathsf{U}}
\newcommand{\ltlX}{\mathsf{X}}

\DeclareMathOperator{\LTL}{LTL}
\DeclareMathOperator{\FOL}{FO}
\newcommand{\LTLFO}{\ensuremath{\LTL^{\hbox{\tiny{$\FOL$}}}\xspace}}

\newcommand{\wrt}{wrt.\xspace}
\newcommand{\st}{s.t.\xspace}

\newcommand{\figref}[1]{Fig.~\ref{#1}}

\newcommand{\calA}{\mathcal{A}}
\newcommand{\calB}{\mathcal{B}}

\newcommand{\calF}{\mathcal{F}}

\newcommand{\calL}{\mathcal{L}}

\newcommand{\sign}{\Gamma}

\newcommand{\struct}{\mathfrak{A}}
\newcommand{\tstruct}{\overline{\struct}}
\newcommand{\dom}{|\struct|}
\newcommand{\sorts}{\mathbf{S}}

\newcommand{\funcs}{\mathbf{F}}
\newcommand{\Rels}{\mathbf{R}}
\newcommand{\rels}{\mathbf{I}}
\newcommand{\preds}{\mathbf{U}}

\newcommand{\vars}{\mathbf{V}}
\newcommand{\LS}{\mathcal{L}(\Gamma)}

\DeclareMathOperator{\events}{Ev}

\newcommand{\cA}{\mathcal{A}}

\newcommand{\cK}{\mathcal{K}}
\newcommand{\cL}{\mathcal{L}}

\newcommand{\cAphi}{\cA_\varphi}
\newcommand{\cAphiv}{\cA_{\varphi,v}}

\newcommand{\buchi}{B\"uchi\xspace}
\DeclareMathOperator{\cl}{cl}
\DeclareMathOperator{\subff}{sf}
\DeclareMathOperator{\depth}{depth}

\DeclareMathOperator{\Inf}{Inf}

\DeclareMathOperator{\good}{good}
\DeclareMathOperator{\bad}{bad}
\DeclareMathOperator{\AP}{AP}
\DeclareMathOperator{\vt}{\vec{t}}
\DeclareMathOperator{\vc}{\vec{c}}
\DeclareMathOperator{\vd}{\vec{d}}
\DeclareMathOperator{\vx}{\vec{x}}

\newcommand{\ddelta}{\delta_\rightarrow}
\newcommand{\spawn}{\delta_\downarrow}
\newcommand{\subf}{{\subff}|_\forall}

\begin{document}

\title{From propositional to first-order monitoring}

\author{Andreas Bauer$^{1,2}$ \and Jan-Christoph K\"uster$^{1,2}$ \and Gil Vegliach$^{1}$}

\institute{$^1$NICTA%
\thanks{NICTA is funded by the Australian Government
    as represented by the Department of Broadband, Communications and
    the Digital Economy and the Australian Research Council through the
    ICT Centre of Excellence program.} 
  Software Systems Research Group, $^2$Australian National University}

\maketitle

\begin{abstract}
  The main purpose of this paper is to introduce a first-order
  temporal logic, $\LTLFO$, and a corresponding monitor construction
  based on a new type of automaton, called spawning automaton.
  
  Specifically, we show that monitoring a specification in $\LTLFO$
  boils down to an undecidable decision problem.
  The proof of this result revolves around specific ideas on what we
  consider a ``proper'' monitor.  As these ideas are general, we
  outline them first in the setting of standard LTL, before lifting
  them to the 
  setting of first-order logic and $\LTLFO$.
  Although due to the above result one cannot hope to obtain a
  complete monitor for $\LTLFO$, we prove the soundness of our
  automata-based construction and give experimental results from an
  implementation.
  %
  These seem to substantiate our hypothesis that the automata-based
  construction leads to efficient runtime monitors whose size does not
  grow with increasing trace lengths (as is often observed in similar
  approaches).  However, we also discuss formulae for which growth is
  unavoidable, irrespective of the chosen monitoring approach.
\end{abstract}

\abovedisplayskip=0.2cm
\belowdisplayskip=0.2cm


\section{Introduction}
\label{sec:intro}
In the area of runtime verification (cf.\
\cite{havelund2004,Halle:2008:RMM:1437901.1438836,Dong:2008:IAR:1478278.1478320,bauer:leucker:schallhart:tosem}),
a \emph{monitor} typically describes a device or program which is
automatically generated from a formal specification capturing undesired
(resp.\ desired) system behaviour.  The monitor's task is to passively
observe a running system in order to detect if the behavioural specification
has been satisfied or violated by the observed system behaviour.
While, arguably, the majority of runtime verification approaches are
based on propositional logic, there exist works that also consider
first-order logic (cf.\
\cite{Halle:2008:RMM:1437901.1438836,bauer:gore:tiu:ictac09,Basin_etal:monitoring_mfotl}).
Monitoring first-order specifications has also gained prior attention in
the database community, especially in the context of so called temporal
triggers, which correspond to first-order temporal logic specifications
that are evaluated \wrt a linear sequence of database updates (cf.\
\cite{DBLP:journals/tods/Chomicki95,DBLP:journals/jcss/ChomickiN95,DBLP:journals/tkde/SistlaW95}).
Although the underlying logics are generally undecidable, the monitors
in these works usually address decidable problems, such as ``is the
observed behaviour so far a violation of a given specification
$\varphi$?''  Additionally, in many approaches, $\varphi$ must only ever
be a safety or domain independent property for this problem to actually
be decidable (cf.\
\cite{DBLP:journals/tods/Chomicki95,Basin_etal:monitoring_mfotl}), which
can be ensured by syntactic restrictions on the input formula, for
example.

As there exist many different ways in which a system can be monitored in
this abstract sense,
we are going to put forth very specific assumptions concerning the
properties and inner-workings of what we consider a ``proper'' monitor.
None of these assumptions is particularly novel or complicated, but they
help describe and distinguish the task of a ``proper'' monitor from that
of, say, a model checker, which can also be used to solve monitoring
problems as we shall see.

The two basic assumptions are easy to explain: Firstly, we demand that a
monitor is what we call \emph{trace-length independent}, meaning that
its efficiency does not decline with an increasing number of
observations.  Secondly, we demand that a monitor is \emph{monotonic}
\wrt reporting violations (resp.\ satisfication) of a specification,
meaning that once the monitor returns ``\textsf{SAT}'' to the user,
additional observations do not lead to it returning ``\textsf{UNSAT}''
(and vice versa).
We are going to postulate further assumptions, but these are mere
consequences of the two basic ones, and are explained in 
\S\ref{sec:prop_case}.

At the heart of this paper, however, is a custom first-order temporal
logic, in the following referred to as \LTLFO, which is undecidable.
%
Yet we outline a sound, albeit incomplete,
monitor construction for it based on a new type of automaton, called
spawning automaton.
$\LTLFO$ was originally developed for the specification of runtime
verification properties of Android ``Apps'' and has already been used in
that context (see \cite{bauer:kuester:vegliach:NFM12} for details).
%
%
%
Although \cite{bauer:kuester:vegliach:NFM12} gave a monitoring algorithm
for $\LTLFO$ based on formula rewriting, it turns out that the
automata-based construction given in this paper leads to practically
more efficient results. 

As our definition of what constitutes a ``proper'' monitor is not tied
to a particular logic we will develop it first for standard LTL
(\S\ref{sec:prop_case}), the quasi-standard in the area of runtime
verification.
%
In \S\ref{sec:ltlfo}, we give a more detailed account of $\LTLFO$ than
was available in \cite{bauer:kuester:vegliach:NFM12}, before we lift the
results of \S\ref{sec:prop_case} to the first-order setting
(\S\ref{sec:fo_case}).  The automata-based monitor construction for
$\LTLFO$ along with experimental results is described in
\S\ref{sec:mon}.
Related work is discussed in \S\ref{sec:related}.
Detailed proofs can be found in a separate appendix.



\section{Complexity of monitoring in the propositional case}
\label{sec:prop_case}

In what follows, we assume basic familiarity with $\LTL$ and topics like
model checking (cf.\ \cite{Baier:2008:PMC:1373322} for an overview).
Despite that, let us first state a formal $\LTL$ semantics, since we
will consider its interpretation on infinite and finite traces.  For
that purpose, let $\AP$ denote a set of propositions, $\LTL(\AP)$ the
set of well-formed LTL formulae over that set, and for some set $X$ set
$X^\infty = X^\omega \cup X^\ast$ to be the union of the set of all
infinite and finite traces over $X$.  When $\AP$ is clear from the
context, or does not matter, we use $\LTL$ instead of $\LTL(\AP)$.
Also, for a given trace $w = w_0 w_1 \ldots$, the trace $w^i$ is defined
as $w_iw_{i+1}\ldots$.
As a convention we use $u, u', \ldots$ to denote finite traces, by
$\sigma$ the trace of length $1$, and $w$ for infinite ones or where the
distinction is of no relevance.

\begin{definition}
  Let $\varphi \in \LTL(\AP)$, $w \in (2^{\AP})^\infty$ be a non-empty
  trace, and $i \in \mathbb{N}_0$, then
  \[
  \begin{array}{rcl}
    w^i \models p & \hbox{ iff } & p \in w_i, \hbox{ where } p \in \AP,\\
    w^i \models \neg\varphi & \hbox{ iff } & w^i \models \varphi \hbox{ does not hold},\\
    w^i \models \varphi \wedge \psi &\hbox{ iff } & w^i \models \varphi \hbox{ and } w^i \models \psi,\\
    w^i \models \ltlX\varphi & \hbox{ iff } & |w| > i \hbox{ and } w^{i+1} \models \varphi,\\
    w^i \models \varphi \ltlU \psi &\hbox{ iff } & \hbox{there is a } k \hbox{ s.t. } i \leq k \leq |w|, w^k \models \psi,
                    \hbox{ and for all } i \leq j < k, w^j \models \varphi.
  \end{array}
  \]
\end{definition}
And if $w^0 \models \varphi$ holds, we usually write $w \models \varphi$
instead.
Although this semantics, which was also proposed in
\cite{DBLP:conf/concur/MarkeyS03}, gives rise to mixed languages, i.e.,
languages consisting of finite and infinite traces, we shall only ever
be concerning ourselves with either finite-trace or infinite-trace
languages, but not mixed ones.  It is easy to see that over infinite
traces this semantics matches the definition of standard $\LTL$.
Recall, $\LTL$ is a decidable logic; in fact, the satisfiability problem
for $\LTL$ is known to be PSpace-complete
\cite{DBLP:journals/jacm/SistlaC85}.

As there are no commonly accepted rules for what 
qualifies as a monitor (not even in the runtime verification community),
there exist a myriad of different approaches to checking that an
observed 
behaviour satisfies (resp.\ violates) a formal specification, such as an
$\LTL$ formula.  Some of these (cf.\
\cite{havelund2004,bauer:gore:tiu:ictac09}) consist in solving the word
problem (see Definition~\ref{def:ltl:wp}).  A monitor following this
idea can either first record the entire system behaviour in
form of a trace $u \in \Sigma^+$, where $\Sigma$ is a finite alphabet of
events, or process the events incrementally as they are emitted by the
system under scrutiny.  Both approaches are documented in the literature
(cf.\
\cite{havelund2004,DBLP:conf/fm/GenonMM06,Halle:2008:RMM:1437901.1438836,bauer:gore:tiu:ictac09}),
but only the second one is suitable to detect property violations
(resp.\ satisfaction) right when they occur.
%

\begin{definition}
  \label{def:ltl:wp}
  The \emph{word problem for $\LTL$} is defined as follows.
  
  \noindent
  \textbf{Input:} A formula $\varphi \in \LTL(\AP)$ and some trace $u \in
  (2^{\AP})^+$.

  \noindent
  \textbf{Question:} Does $u \models \varphi$ hold?
\end{definition}
In \cite{DBLP:conf/concur/MarkeyS03} a bilinear algorithm for this
problem was presented (an even more efficient solution was recently
given in \cite{DBLP:journals/corr/abs-1210-0574}).
%
%
Hence, the first sort of monitor, which is really more of a test oracle
than a monitor, solves a classical decision problem.  The second
monitor, however, solves an entirely different kind of problem, which
cannot be stated in complexity-theoretical terms at all: its input is an
$\LTL$ formula and a finite albeit unbounded trace which \emph{grows}
incrementally.  This means that this monitor solves the word problem for
each and every new event that is added to the trace at runtime.  We can
therefore say 
that the word problem acts as a lower bound on the complexity of the
monitoring problem that such a monitor solves;
or, in other words, the problem that the online monitor solves is at
least as hard as the problem that the offline monitor solves.

There are approaches to build efficient (i.e., trace-length
independent) monitors that repeatedly answer the word problem (cf.\
\cite{havelund2004}).  However, such approaches violate our second basic
assumption, mentioned in the introduction of this paper, in that they
are necessarily non-monotonic.  To see this, consider $\varphi = a \ltlU
b$ and some trace $u = \{ a \} \{ a \} \ldots \{ a \}$ of length $n$.
Using our finite-trace interpretation, $u \not\models \varphi$.
However, if we add $u_{n+1} = \{b\}$, we get $u \models
\varphi$.\footnote{Note that this effect is not particular to our choice
  of finite-trace interpretation.  Had we used, e.g., what is known as
  the weak finite-trace semantics, discussed in
  \cite{DBLP:conf/cav/EisnerFHLMC03}, we would first have had $u \models
  \varphi$ and if $u_{n+1} = \emptyset$, subsequently $u
  \not\models\varphi$.}  For the user, this essentially means that she
cannot trust the verdict of the monitor as it may flip in the future,
unless of course it is obvious from the start that, e.g., only safety
properties are monitored and the monitor is built merely to detect
violations, i.e., bad prefixes.
However, if we take other monitorable languages into account as we do in
this paper, i.e., those that have either good or bad prefixes (or both),
we need to distinguish between satisfaction and violation of a property
(and want the monitor to report either occurrence truthfully).

\begin{definition}
  \label{def:prefixes}
  For any $L \subseteq \Sigma^\omega$, $u\in\Sigma^\ast$ is called a
  \emph{good prefix} (resp.\ \emph{bad prefix}) iff $u\Sigma^\omega
  \subseteq L$ holds (resp.\ $u\Sigma^\omega \cap L = \emptyset$).  
\end{definition}
We shall use $\good(L) \subseteq \Sigma^\ast$ (resp.\ $\bad(L)$) to
denote the set of good (resp.\ bad) prefixes of $L$.  For brevity, we
also write $\good(\varphi)$ instead of $\good(\calL(\varphi))$, and do
the same for $\bad(\calL(\varphi))$.

A monitor that detects good (resp.\ bad) prefixes has been termed
anticipatory in \cite{Dong:2008:IAR:1478278.1478320} as it not only
states something about the past, but also about the future: once a good
(resp.\ bad) prefix has been detected, no matter how the system would
evolve in an indefinite future, the property would remain satisfied
(resp.\ violated).  In that sense, anticipatory monitors are monotonic
by definition.  Moreover in \cite{bauer:leucker:schallhart:tosem}, a
construction is given, showing how to obtain trace-length independent
(even optimal) anticipatory monitors for $\LTL$ and a timed extension
called TLTL.  The obtained monitor basically returns $\top$ to the user
if $u \in \good(\varphi)$ holds, $\bot$ if $u \in \bad(\varphi)$ holds,
and $?$ otherwise.  Not surprisingly though, the monitoring problem such
a monitor solves is computationally more involved than the word problem.
It solves what we call the prefix problem (of $\LTL$), which can easily
be shown PSpace-complete by way of $\LTL$ satisfiability.

\begin{definition} \label{def:ltl:pp}
  The \emph{prefix problem for $\LTL$} is defined as follows.
  
  \noindent
  \textbf{Input:} A formula $\varphi \in \LTL(\AP)$ and some trace $u \in
  (2^{\AP})^\ast$.

  \noindent
  \textbf{Question:} Does $u \in \good(\varphi)$ (resp.\ $\bad(\varphi)$) hold?
\end{definition}

\begin{theorem}
  \label{thm:prop:prefix}
  The prefix problem for $\LTL$ is PSpace-complete.
\end{theorem}
\begin{proof}
  For brevity, we will only show the theorem for bad prefixes.
  It is easy to see that $u \in \bad(\varphi)$ iff $\calL(u_0 \wedge
  \ltlX u_1 \wedge \ltlX\ltlX u_2 \wedge \ldots \wedge \varphi) =
  \emptyset$.  Constructing this conjunction takes only polynomial time
  and the corresponding emptiness check can be performed in PSpace
  \cite{DBLP:journals/jacm/SistlaC85}.
  To show hardness, we proceed with a reduction of LTL satisfiability.
  Again, it is easy to see that $\calL(\varphi) \neq \emptyset$ iff
  $\sigma \not\in \bad(\ltlX\varphi)$ for any $\sigma \in 2^{\AP}$.  This
  reduction is linear, and as PSpace $=$ co-PSpace, the statement
  follows. \qed
\end{proof}

We would like to point out the possibility of building an
anticipatory though trace-length dependent $\LTL$ monitor using an ``off
the shelf'' model checker, which accepts a propositional Kripke
structure and an $\LTL$ formula as input.
Note that here we make the 
assumption that Kripke structures produce infinite as opposed to finite
traces.

\begin{definition}
  The \emph{model checking problem for $\LTL$} is defined as follows.
  
  \noindent
  \textbf{Input:} A formula $\varphi \in \LTL(\AP)$ and a Kripke
  structure $\cK$ over $2^{\AP}$.

  \noindent
  \textbf{Question:} Does $\calL(\cK) \subseteq \calL(\varphi)$ hold?
\end{definition}
As in $\LTL$ the model checking and the satisfiability problems are both
PSpace-complete \cite{DBLP:journals/jacm/SistlaC85}, we can use a model
checking tool as monitor: given that it is straightforward to construct
$\cK$ \st $\cL(\cK) = u(2^{\AP})^\omega$ in no more than polynomial
time, we return $\top$ to the user if $\cL(\cK) \subseteq \cL(\varphi)$
holds, $\bot$ if $\cL(\cK) \subseteq \cL(\neg\varphi)$ holds, and $?$ if
neither holds.
One could therefore be tempted to think of monitoring merely in terms of
a 
model checking problem, but we shall see that as soon as the logic in
question has an undecidable satisfiability problem this reduction
fails. 
Besides, it can be questioned whether monitoring as model checking leads
to a desirable 
monitor with its obvious trace-length dependence and having to
repeatedly solve a PSpace-complete problem for each new event.



\section{\LTLFO---Formal definitions and notation}
\label{sec:ltlfo}

Let us now introduce our first-order specification language $\LTLFO$ and
related concepts in more detail.
The first concept we need is that of a \emph{sorted first-order
  signature}, given as $\sign = (\sorts, 
\funcs, \allowbreak \Rels)$, where $\sorts$ is a finite non-empty set
of sorts, $\funcs$ a finite set of function symbols and $\Rels =
\preds \cup \rels$ a finite set of a priori uninterpreted and
interpreted predicate symbols, \st $\preds \cap \rels = \emptyset$ and
$\Rels \cap \funcs = \emptyset$.  The former set of predicate symbols
are referred to as $\preds$-operators and the latter as
$\rels$-operators. As is common, 0-ary functions symbols are also
referred to as constant symbols.
%
We assume that all operators in $\sign$ have a given arity that ranges
over the sorts given by $\sorts$, respectively.
We also assume an infinite supply of variables, $\vars$, that also range
over $\sorts$ and where $\vars \cap (\funcs \cup \Rels) =
\emptyset$.
Let us refer to the first-order language determined by $\sign$ as $\LS$.
While \emph{terms} in $\LS$ are made up of variables and function
symbols,
\emph{formulae} of $\LS$ are defined as follows:
$$
\varphi ::= p(t_1, \ldots, t_n) \mid 
r(t_1, \ldots, t_n) \mid 
\neg \varphi \mid
\varphi \wedge \varphi \mid
\ltlX\varphi \mid
\varphi \ltlU \varphi \mid
\forall (x_1, \ldots, x_n):p.\ \varphi,
$$
where $t_1, \ldots, t_n$ are terms, $p \in \preds$, $r \in \rels$, and
$x_1, \ldots, x_n \in \vars$.  
As variables are sorted, in the quantified formula $\forall (x_1,
\ldots, x_n):p.\ \varphi$, the $\preds$-operator $p$ with arity $\tau_1
\times \ldots \times \tau_{n}$, defines the sorts of variables $x_1,
\ldots, x_n$ to be $\tau_1, \ldots, \tau_{n}$, with $\tau_i \in \sorts$,
respectively.
For terms $t_1,\ldots,t_{n}$, we say that $p(t_1,\ldots,t_{n})$ is
well-sorted if the sort of every $t_i$ is $\tau_i$.
This notion is inductively applicable to terms.  Moreover, we consider
only well-sorted formulae and refer to the set of all well-sorted $\LS$
formulae over a signature $\sign$ in terms of $\LTLFO_\sign$.
When a specific $\sign$ is either irrelevant or clear from the context,
we will simply write $\LTLFO$ instead.
When convenient and a certain index is of no importance in the given
context, we also shorten notation of a vector $(x_1, \ldots, x_n)$ by a
(bold) $\vx$.

A $\Gamma$\emph{-structure}, or just \emph{first-order structure} is a
pair $\struct = (\dom, I)$, where
%
%
$\dom = \dom_1 \cup \ldots \cup \dom_n$,
is a non-empty set called domain, \st every sub-domain $\dom_i$ is
either a non-empty finite or countable set (e.g., set of all integers or
strings) and $I$ an interpretation.
$I$ assigns to each sort $\tau_i \in \sorts$ a specific sub-domain
$\tau_i^I = \dom_i$, 
to each constant symbol $c \in \funcs$ of sort $\tau_i$ a domain value
$c^I \in \dom_i$, to each function symbol $f \in \funcs$ of arity
$\tau_1 \times \ldots \times \tau_l \longrightarrow \tau_m$
a function $f^I: \dom_1 \times \ldots \times \dom_l \longrightarrow
\dom_m$, and to every $\rels$-operator $r$ with arity $\tau_1 \times
\ldots \times \tau_m$
a relation $r^I \subseteq \dom_1 \times \ldots \times \dom_m$.
%
We restrict ourselves to computable relations and functions.
In that regard, we can think of $I$ as a mapping between
$\rels$-operators (resp.\ function symbols) and the corresponding
algorithms which compute the desired return values, each conforming to
the symbols' respective arities.
Note that the interpretation of $\preds$-operators is rather different
from $\rels$-operators, as it is closely tied to what we call a trace
and therefore discussed in more detail after we introduce the necessary
notions and notation.
%

%

For the purpose of monitoring $\LTLFO$ specifications, we model observed
system behaviour in terms of \emph{actions}: Let $p \in \preds$ with
arity $\tau_1 \times \ldots \times \tau_m$ and $\vd \in D_p = \dom_1
\times \ldots \times \dom_m$, then we call $(p, \vd)$ an \emph{action}.
We refer to finite sets of actions as \emph{events}.
A system's behaviour over time is therefore a finite \emph{trace} of
events, which we also denote as a sequence of sets of ground terms $\{
sms(1234) \} \{ login(``\hbox{user}") \} \ldots$ when we mean the
sequence of tuples $\{ (sms, 1234) \}\allowbreak \{(login,
``\hbox{user}") \} \ldots$
Therefore the occurrence of some action $sms(1234)$ in the trace at
position $i \in \mathbb{N}_0$, written $sms(1234) \in w_i$, indicates
that at time $i$ it is the case that $sms(1234)$ holds (or, from a
practical point of view, an SMS was sent to number $1234$).  
We follow the convention that only symbols from $\preds$ appear in a
trace, which therefore gives these symbols their respective
interpretations.  The following formalises this notion.

A \emph{first-order temporal structure} is a tuple $(\tstruct, w)$,
where $\tstruct = (|\struct_0|, I_0) \allowbreak (|\struct_1|,
I_1)\allowbreak \ldots$ is a (possibly infinite) sequence of
first-order structures and $w = w_0w_1\ldots$ a corresponding trace.
We demand that for all $\struct_i$ and $\struct_{i+1}$ from $\tstruct$,
it is the case that $|\struct_i| = |\struct_{i+1}|$, for all $f \in
\funcs$, $f^{I_{i+1}} = f^{I_i}$, and for all $\tau \in \sorts$,
$\tau^{I_{i}} = \tau^{I_{i+1}}$.
%
For any two structures, $\struct$ and $\struct'$, which satisfy these
conditions, we write $\struct \sim \struct'$.
%
%
Moreover given some $\tstruct$ and $\struct$, if for
all $\struct_i$ from $\tstruct$, we have that $\struct_i \sim
\struct$, we also write $\tstruct \sim \struct$.
Finally, the interpretation of an $\preds$-operator $p$ with arity
$\tau_1 \times \ldots \times \tau_m$ is then defined \wrt a position $i$
in $w$ as $p^{I_i} = \{ \vd \mid (p, \vd) \in w_i \}$.  Essentially this
means that, unlike function symbols, $\preds$- and $\rels$-operators
don't have to be rigid.

Note also that from this point forward, we consider only the case where
the policy to be monitored is given as a closed formula, i.e., a
sentence.  This is closely related to our means of quantification:
a quantifier in $\LTLFO$ is restricted to those elements that appear in
the trace, and not arbitrary elements from a (possibly infinite) domain.
While certain policies cannot be expressed with this restriction (e.g.,
``for all phone numbers $x$ that are not in the contact list, $r(x)$ is
true''), this restriction bears the advantage that, when examining a
given trace, functions and relations are only ever evaluated over known
objects.
The advantages of this type of quantification in monitoring first-order
languages has also been pointed out in
\cite{Halle:2008:RMM:1437901.1438836,bauer:gore:tiu:ictac09}.
In other words, had we allowed free variables in policies, the monitor
might end up having to ``try out'' all the different domain elements in
order to evaluate such a policy, which runs counter to our design
rationale of quantification.

In what follows, let us fix a particular $\Gamma$. The semantics of
$\LTLFO$ can now be defined \wrt a quadruple $(\tstruct, w, v, i)$ as
follows, where $i \in \mathbb{N}_0$, and $v$ is an (initially empty) set
of valuations assigning domain values to variables:
{\allowdisplaybreaks
\addtolength{\jot}{-0.3em}
\begin{align*}
  (\tstruct, w, v, i) \models p(t_1, \ldots, t_n) & \hbox{~ iff ~}  (t_1^{I_i}, \ldots, t_n^{I_i}) \in p^{I_i}, \\
  (\tstruct, w, v, i) \models r(t_1, \ldots, t_n) & \hbox{~ iff ~}  (t_1^{I_i}, \ldots, t_n^{I_i}) \in r^{I_i} , \\ 
  (\tstruct, w, v, i) \models \neg \varphi       & \hbox{~ iff ~} (\tstruct, w, v, i) \models \varphi \hbox{ is not true},\\
  (\tstruct, w, v, i) \models \varphi \wedge \psi & \hbox{~ iff ~} (\tstruct, w, v, i) \models \varphi \hbox{ and } (\tstruct, w, v, i) \models \psi, \\
  (\tstruct, w, v, i) \models \ltlX\varphi       & \hbox{~ iff ~}  |w| > i \hbox{ and } (\tstruct, w, v, i + 1) \models \varphi,\\
  (\tstruct, w, v, i) \models \varphi \ltlU\psi  &\hbox{~ iff ~}  \hbox{for some } k \geq i, (\tstruct, w, v, k) \models \psi,\\
  & { ~ ~ }  \hbox{and }  (\tstruct, w, v, j) \models \varphi \hbox{ for all } i \leq j < k,\\
  (\tstruct, w, v, i) \models \forall (x_1, \ldots, x_n):p.\ \varphi & \hbox{~ iff ~} 
  \hbox{for all } (p, d_1, \ldots, d_n) \in w_i, \\
  & { ~ ~ }   (\tstruct, w, v \cup \{x_1\mapsto d_1, \ldots, x_n\mapsto d_n\}, i) \models \varphi,
\end{align*}}%
where terms are evaluated inductively, and $x^I$ means $v(x)$.
If $(\tstruct, w, v, 0) \models \varphi$, we write $(\tstruct, w, v)
\models \varphi$, and if a particular $v$ is irrelevant or clear from
the context, we shortcut the latter simply to $(\tstruct, w) \models
\varphi$.

Later we will also make use of the (possibly infinite) set of all events
\wrt $\struct$, given as
$(\struct)$-$\events = \bigcup_{p \in \preds} \{ (p, \vd) \mid \vd \in D_p \}$,
and take the liberty to omit the trailing $(\struct)$ whenever a
particular $\struct$ is either irrelevant or clear from the context.
We can then describe the \emph{generated language} of $\varphi$,
$\calL(\varphi)$ (or simply the language of $\varphi$, i.e., the set of
all logical models of $\varphi$) compactly as
$
\calL(\varphi) = \{ (\tstruct, w) \mid w_i \in 2^{\events} \hbox{ and } (\tstruct, w) \models \varphi \},
$ 
although, as before, we shall only ever concern ourselves with either
infinite- or finite-word languages, but not mixed ones.
Finally, we will use common syntactic ``sugar'', including $\exists
(x_1, \ldots, x_n):p.\ \varphi = \neg (\forall (x_1, \ldots, x_n):p.\
\neg \varphi)$, etc.


For brevity, we refer the reader to \cite{bauer:kuester:vegliach:NFM12}
for some example policies formalised in $\LTLFO$.  However, to give at
least an intuition, let's pick up the idea of monitoring Android
``Apps'' again,
and specify that 
``Apps'' must not send SMS messages to numbers not in a user's contact
database.  Assuming there exists an $\preds$-operator $sms$, which is
true / appears in the trace, whenever an ``App'' sends an SMS message to
phone number $x$, we could formalise said policy in terms of $\ltlG
\forall x:sms.\ contact(x)$.  Note how in this formula the meaning of
$x$ is given implicitly by the arity of $sms$ and must match the
definition of $contact$. 
Also note how $sms$ is interpreted indirectly via its occurrence in the
trace, whereas $contact$ never appears in the trace, even if true.
$contact$ can be thought of as interpreted via a program that queries a
user's contact database, whose contents may change over time.



\section{Complexity of monitoring in the first-order case}
\label{sec:fo_case}

$\LTLFO$ as defined above is undecidable as can be
shown by way of the following lemma whose detailed proof is available in
the appendix.  It basically helps us reduce finite satisfiability of
standard first-order logic to $\LTLFO$.

\begin{lemma}
  \label{lem:finmodel}
  Let $\varphi$ be a sentence in first-order logic, then we can
  construct a corresponding $\psi \in \LTLFO$ \st $\varphi$ has a finite
  model iff $\psi$ is satisfiable.
\end{lemma}

\begin{theorem}
  \label{thm:fo:undecidable}
  $\LTLFO$ is undecidable.
\end{theorem}
\begin{proof}[Idea]
  Follows from Lemma~\ref{lem:finmodel} and Trakhtenbrot's Theorem (cf.\
  \cite[\S9]{Libkin:2004:EFM:1024196}).
\end{proof}

Let us now define what we mean by Kripke structure in our new setting,
and the generated language of it.  
The Kripke structures we consider either give rise to infinite languages
(i.e., have a left-total transition relation), or represent traces (i.e,
are essentially linear structures).
For brevity, we shall restrict to the definition of the former.  Note
that we will also skip detailed redefinitions of the decision problems
discussed in \S\ref{sec:prop_case}, since the employed concepts transfer
in a straightforward manner.

\begin{definition}
  \label{def:kripke}
  Given some $\struct$, a \emph{($\struct$)-Kripke structure}, or just
  first-order Kripke structure, is a state-transition system $\cK = (S,
  s_0, \lambda, \rightarrow)$, where $S$ is a finite set of states, $s_0
  \in S$ a distinguished initial state, $\lambda: S \longrightarrow
  \widehat{\struct} \times \events$, where $\widehat{\struct}=\{\struct'
  \mid \struct' \sim \struct\}$, a labelling function, and $\rightarrow
  \subseteq S \times S$ a (left-total) transition relation.
\end{definition}
\begin{definition}
  \label{def:kripke-language}
  For a $(\struct)$-Kripke structure $\cK$ with states $s_0, \ldots,
  s_n$, the \emph{generated language} is given as
  $\calL(\cK) = \{ (\tstruct, w) \mid 
    (\struct_0, w_0) = \lambda(s_0) 
    \hbox{ and for all } i \in \mathbb{N} \hbox{ there exist } \allowbreak \hbox{some } j,k \in [ 0, n ] 
    \hbox{ s.t.\ } (\struct_i, w_i) = \lambda(s_j), (\struct_{i-1}, w_{i - 1}) = \lambda(s_{k}) \hbox{ and } 
    (s_{k}, s_j) \in \rightarrow \}.$
\end{definition}

The inputs to the $\LTLFO$ word problem are therefore an $\LTLFO$
formula and a linear first-order Kripke structure, representing a finite
input trace.  Unlike in standard $\LTL$, we note that

\begin{theorem}
  \label{lem:fo:word}
  The word problem for $\LTLFO$ is PSpace-complete.
\end{theorem}
%
%
The inputs to the $\LTLFO$ model checking problem, in turn, are a
left-total first-order Kripke structure, which gives rise to an
infinite-trace language, and an $\LTLFO$ formula.
\begin{theorem}
  \label{lem:fo:mc}  
  The model checking problem for $\LTLFO$ is in
  ExpSpace.
\end{theorem}
%
%
The reason for this result is that we can devise a reduction of the
$\LTLFO$ model checking problem to LTL model checking, using
exponential space.
While it is easy to obtain a PSpace-lower bound, for example via a
reduction of the $\LTLFO$ word problem, we currently do not know how
tight these bounds are and, therefore, leave this as an open problem.
Note also that the results of both Theorem~\ref{lem:fo:word} and
Theorem~\ref{lem:fo:mc} are obtained even without taking into account
the complexities of the interpretations of function symbols and
$\rels$-operators; that is, for these results to hold, we assume that
interpretations do not exceed polynomial, resp.\ exponential space.

We have seen in \S\ref{sec:prop_case} that the prefix problem lies at
the heart of an anticipatory monitor.  While in $\LTL$ it was possible
to build an anticipatory monitor using a model checker (albeit a very
inefficient one), Theorem~\ref{lem:fo:pf} shows that this is no longer
possible for $\LTLFO$.  Its proof makes use of the following
intermediate lemma.
\begin{lemma}
  \label{lem:fo:restricted}
  Let 
  $\struct$ be a first-order
  structure and $\varphi \in \LTLFO$, then
  $\cL(\varphi)_{\struct} = 
  \{ (\tstruct, w) \mid \tstruct \sim \struct, w \in (2^{\events})^\omega,
  \hbox{ and } (\tstruct, w) \models \varphi \}$.
  %
  Testing if $\cL(\varphi)_{\struct} \neq \emptyset$ is generally
  undecidable.
\end{lemma}
\begin{proof}[Idea]
  By a reduction from Post's Correspondence Problem.
\end{proof}

\begin{theorem}
  \label{lem:fo:pf}  
  The prefix problem for $\LTLFO$ is undecidable.
\end{theorem}
\begin{proof}[Idea]
  Similar to Theorem~\ref{thm:prop:prefix}:
  $(\struct, \sigma) \in \bad(\ltlX\varphi)$ iff $\calL(\varphi)_\struct
  = \emptyset$ for any $\sigma \in \events$.
\end{proof}



\section{Monitoring \LTLFO}
\label{sec:mon}

A direct consequence of Theorem~\ref{lem:fo:pf} is that there cannot
exist a complete monitor for $\LTLFO$-definable infinite trace
languages.  Yet one of the main contributions of our work is to show
that one can build a sound and efficient $\LTLFO$ monitor using a new
kind of automaton.  Before we go into the details of the actual
monitoring algorithm, let us first consider the automaton model, which
we refer to as \emph{spawning automaton} (SA).
SAs are called that, because when they process their input, they
potentially ``spawn'' a positive Boolean combination of ``children SAs''
(i.e., subautomata) in each such step.
Let $\calB^+(X)$ denote the set of all positive Boolean formulae over
the set $X$.  We say that some set $Y \subseteq X$ satisfies a formula
$\beta \in \calB^+(X)$, written $Y \models \beta$, if the truth
assignment that assigns true to all elements in $Y$ and false to all $X
- Y$ satisfies $\beta$.

\begin{definition}
  A \emph{spawning automaton}, or simply SA, is given by $\cA = (\Sigma,
  l, Q, Q_0, \ddelta, \allowbreak \spawn, \allowbreak \calF)$,
  where $\Sigma$ is a countable set called alphabet, $l \in
  \mathbb{N}_0$ the level of $\cA$, $Q$ a finite set of states, $Q_0
  \subseteq Q$ a set of distinguished initial states, $\ddelta$ a
  transition relation, $\spawn$ what is called a spawning function, and
  $\calF = \{ F_1, \ldots, F_n \mid F_i \subseteq Q\}$ an acceptance
  condition (to be defined later on).
  $\ddelta$ is given as $\ddelta: Q \times \Sigma \longrightarrow 2^Q$.
  The spawning function $\spawn$ is then given as $\spawn: Q \times
  \Sigma \longrightarrow \calB^+( \cA^{<l})$, where $\cA^{<l} = \{ \cA'
  \mid \cA' \hbox{ is an SA with level less than } l \}$.
\end{definition}

\begin{definition}  
  \label{def:run}
  A \emph{run} of 
  $\cA= (\Sigma, l, Q, Q_0, \ddelta, \spawn, \calF)$ over an input
  sequence $w \in \Sigma^\omega$ is a mapping $\rho: \mathbb{N}_0
  \longrightarrow Q$, \st $\rho(0) \in Q_0$ and $\rho(i + 1) \in
  \ddelta(\rho(i), w_i)$ for all $i \in \mathbb{N}_0$.
  $\rho$ is \emph{locally accepting} if $\Inf(\rho) \cap F_i \neq
  \emptyset$ for all $F_i \in \calF$, where $\Inf(\rho)$ denotes the set
  of states visited infinitely often.
  It is called \emph{accepting} if $l = 0$ and it is locally
  accepting.  If $l > 0$, $\rho$ is called accepting if it is locally
  accepting and for all $i \in \mathbb{N}_0$ there is a set $Y
  \subseteq \cA^{<l}$, \st $Y \models \spawn(\rho(i), w_i)$ and all
  automata $\cA' \in Y$ have an accepting run, $\rho'$, over $w^i$.
  %
  The \emph{accepted language} of $\cA$, $\calL(\calA)$, consists of
  all $w \in \Sigma^\omega$, for which it has at least one accepting
  run. 
\end{definition}
%



\subsection{Spawning automata construction}

Given some $\varphi \in \LTLFO$, let us now examine in detail how to
build the corresponding SA, $\cAphi = (\Sigma, l, Q, Q_0, \ddelta,
\spawn, \calF)$ \st $\calL(\cAphi) = \calL(\varphi)$ holds.
To this end, we set $\Sigma = \{ (\struct, \sigma) \mid \sigma \in
(\struct)$-$\events \}$.
%
If $\varphi$ is not a sentence, we write $\cAphiv$ to denote the
spawning automaton for $\varphi$ in which free variables are mapped
according to a finite set of valuations $v$.\footnote{Considering free
  variables, even though our runtime policies can only ever be
  sentences, is necessary, because an SA for a policy $\varphi$ is
  inductively defined in terms of SAs for its subformulae (i.e.,
  $\cAphi$'s subautomata), some of which may contain free variables.}
To define the set of states for an SA, we make use of a restricted
subformula function, $\subf(\varphi)$, which is defined like a generic
subformula function, except if $\varphi$ is of the form $\forall \vx:p.\
\psi$, we have $\subf(\varphi) = \{ \varphi \}$.
This essentially means that an SA for a formula $\varphi$ on the topmost
level looks like the \buchi automaton (BA, cf.\
\cite{Baier:2008:PMC:1373322}) for $\varphi$, where quantified
subformulae have been interpreted as atomic propositions.

For example, if $\varphi = \psi \wedge \forall \vx:p.\ \psi'$, where
$\psi$ is a quantifier-free formula, then $\cAphi$, at the topmost level
$n$, is like the BA for the LTL formula $\psi \wedge a$, where $a$ is an
atomic proposition; or in other words, $\cAphi$ handles the subformula
$\forall \vx:p.\ \psi'$ separately in terms of a subautomaton of level
$n - 1$ (see also definition of $\spawn$ below).

%
Finally, we define the closure of $\varphi$ \wrt $\subf(\varphi)$ as
$\cl(\varphi) = \{ \neg\psi \mid \psi \in \subf(\varphi) \} \cup \subf(\varphi)$,
%
i.e., the
smallest set containing $\subf(\varphi)$, which is closed under
negation.  

The \emph{set of states} of $\cAphi$, $Q$, consists of all complete
subsets of $\cl(\varphi)$;
that is, a set $q \subseteq \cl(\varphi)$ is complete iff
\begin{itemize}
\item for any $\psi \in \cl(\varphi)$ either $\psi \in q$ or $\neg\psi
  \in q$, but not both; and
\item for any $\psi \wedge \psi' \in \cl(\varphi)$, we have that $\psi
  \wedge \psi' \in q$ iff $\psi \in q$ and $\psi' \in q$; and
\item for any $\psi \ltlU \psi' \in \cl(\varphi)$, we have that if $\psi
  \ltlU \psi' \in q$ then $\psi' \in q$ or $\psi \in q$, and if $\psi
  \ltlU \psi' \not\in q$, then $\psi' \not\in q$.
\end{itemize}
Let $q \in Q$ and $\struct = (\dom, I)$. 
The \emph{transition function} $\ddelta(q, (\struct, \sigma))$ is defined
iff
\begin{itemize}
\item for all $p(\vt) \in q$, we have $\vt^I \in p^I$ and for all $\neg p(\vt) \in q$, we have $\vt^I \not\in p^I$,
\item for all $r(\vt) \in q$, we have $\vt^I \in r^I$ and for all $\neg r(\vt) \in q$, we have $\vt^I \not\in r^I$.
\end{itemize}
In which case, for any $q' \in Q$, we have that $q' \in \ddelta(q,
(\struct, \sigma))$ iff
\begin{itemize}
\item for all $\ltlX\psi \in \cl(\varphi)$, we have $\ltlX\psi \in q$ iff $\psi \in q'$, and
\item for all $\psi \ltlU \psi' \in \cl(\varphi)$, we have $\psi \ltlU \psi' \in
  q$ iff $\psi' \in q$ or $\psi \in q$ and $\psi \ltlU \psi' \in q'$.
\end{itemize}
This is similar to the well known syntax directed construction of BAs
(cf.\ \cite{Baier:2008:PMC:1373322}), except that we also need to cater
for quantified subformulae.  For this purpose, an inductive
\emph{spawning function} is defined as follows.  If $l > 0$, then
$\spawn(q, (\struct,\sigma))$ yields
\[
\left( \bigwedge_{\forall \vec{x}:p. \psi \in q} 
\left( \bigwedge_{(p,\vd) \in \sigma} \cA_{\psi,v'} \right) \right)
\wedge
\left( \bigwedge_{\neg\forall \vec{x}:p. \psi \in q} 
\left( \bigvee_{(p,\vd) \in \sigma} \cA_{\neg\psi,v''} \right) \right),
\]
where $v' = v \cup \{\vx \mapsto \vd\}$ and $v'' = v \cup \{\vx \mapsto
\vd\}$ are sets of valuations, otherwise $\spawn(q, (\struct, \sigma))$
yields $\top$.
%
%
Moreover, we set $Q_0 = \{ q \in Q \mid \varphi \in q \}$, $\calF = \{
F_{\psi \ltlU \psi'} \mid \psi \ltlU \psi' \in \cl(\varphi) \}$ with
$F_{\psi \ltlU \psi'} = \{ q \in Q \mid \psi' \in q \vee
\neg(\psi\ltlU\psi') \in q \}$, and $l = \depth(\varphi)$, where
$\depth(\varphi)$ is called the \emph{quantifier depth} of $\varphi$.
For some $\varphi \in \LTLFO$, $\depth(\varphi) = 0$ iff $\varphi$ is a
quantifier free formula.  The remaining cases are inductively defined as
follows: $\depth(\forall \vx:p.\ \psi) = 1 + \depth(\psi)$, $\depth(\psi
\wedge \psi') = \depth(\psi \ltlU \psi') = \max(\depth(\psi),
\depth(\psi'))$ and $\depth(\neg\varphi) = \depth(\ltlX\varphi) =
\depth(\varphi)$.



\begin{lemma}
  \label{lem:sa:correct}
  Let $\varphi \in \LTLFO$ (not necessarily a sentence) and $v$ be a
  valuation. For each accepting run $\rho$ in $\cAphiv$ over input
  $(\tstruct, w)$, $\psi \in \cl(\varphi)$, and $i \geq 0$, we have that
  $\psi \in \rho(i)$ iff $(\tstruct, w, v, i) \models \psi$.
\end{lemma}
\begin{proof}[Idea]
  By nested induction on $\depth(\varphi)$ and the structure of $\psi
  \in \cl(\varphi)$.
\end{proof}

\begin{theorem}
  \label{thm:sa:correct}
  The constructed SA is correct in the sense that for any sentence
  $\varphi \in \LTLFO$, we have that $\calL(\cA_\varphi) =
  \calL(\varphi)$.
\end{theorem}
\begin{proof}[Idea]
  $\subseteq$ by Lemma~\ref{lem:sa:correct}.  The other direction uses
  induction on $\depth(\varphi)$.
\end{proof}

\subsection{Monitor construction}

Before we look at the actual monitor construction, let us first
introduce some additional concepts and notation:
%
For a finite run $\rho$ in $\cA_{\varphi}$ over $(\tstruct, u)$, we
call $\spawn(\rho(j), (\struct_j, u_j)) = obl_j$ an
\emph{obligation}, where $0 \leq j < |u|$,
in that $obl_j$ represents the language to be satisfied after $j$
inputs.  That is, $obl_j$ refers to the language represented by the
positive Boolean combination of spawned SAs.
%
We say it is \emph{met} by the input, if $(\tstruct^j, u^j) \in
\good(obl_j)$ and \emph{violated} if $(\tstruct^j, u^j) \in
\bad(obl_j)$.
%
Furthermore, $\rho$ is called \emph{potentially locally accepting}, if
it can be extended to a run $\rho'$ over $(\tstruct, u)$ together with
some infinite suffix, such that $\rho'$ is locally accepting.


The monitor for a given formula $\varphi \in \LTLFO$ can now be described
in terms of two mutually recursive algorithms:
The main entry point is Algorithm~M.  It reads an event and issues two
calls to a separate Algorithm~T, one for $\varphi$ (under a possibly
empty valuation $v$) and one for $\neg\varphi$ (under a possibly empty
valuation $v$).
The purpose of Algorithm~T is to detect bad prefixes \wrt the language
of its argument formula, call it $\psi$.  It does so by keeping track of
those finite runs in $\cA_{\psi,v}$ that are potentially locally
accepting and where its obligations haven't been detected as violated by
the input.
If at any time not at least one such run exists, then a bad prefix has
been encountered.
Algorithm~T, in turn, uses Algorithm~M to evaluate if 
obligations of its runs are met or violated by the input observed so far
(i.e., it inductively creates submonitors): after the $i$th input, it
instantiates Algorithm~M with argument $\psi'$ (under corresponding
valuation $v'$) for each $\cA_{\psi',v'}$ that occurs in $obl_i$ and
forwards to it all observed events from time point $i$ on.

\begin{algo}{M}{Monitor}
  The algorithm takes a $\varphi \in \LTLFO$ (under a possibly empty
  valuation $v$).  Its abstract behaviour is as follows:
  Let us assume an initially empty first-order temporal structure
  $(\tstruct, u)$.
  Algorithm~M reads an event $(\struct, \sigma)$, prints ``$\top$'' if
  $(\tstruct\struct, u\sigma) \in \good(\varphi)$ (resp.\ ``$\bot$'' for
  $\bad(\varphi)$), and returns.  Otherwise it prints ``$?$'', whereas
  we now assume that $(\tstruct, u) = (\tstruct\struct,u\sigma)$
  holds.\footnote{Obviously, the monitor does not really keep
    $(\tstruct, u)$ around, or it would be necessarily trace-length
    dependent.  $(\tstruct, u)$ is merely used here to explain the inner
    workings of the monitor.}
  \setlist{nosep}
  \begin{description}[style=multiline,leftmargin=0.73cm]  
  \item [M1.] [Create instances of Algorithm~T.] Create two instances of
    Algorithm~T: one with $\varphi$ and one with $\neg\varphi$, and call
    them $T_{\varphi,v}$ and $T_{\neg\varphi, v}$, respectively.
  \item [M2.\namedlabel{itm:read_event}{M2}] [Forward next event.]
    Wait for next event $(\struct, \sigma)$ and forward it to
    $T_{\varphi,v}$ and $T_{\neg\varphi,v}$.
  \item [M3.\namedlabel{algo:m:verdict}{M3}] [Communicate verdict.] If
    $T_{\varphi,v}$ sends ``no runs'', print $\bot$ and return.  If
    $T_{\neg\varphi,v}$ sends ``no runs'', print $\top$ and
    return. Otherwise, print ``?'' and go to \ref{itm:read_event}.
    \hspace{1em} \Rectangle
  \end{description}
\end{algo}

\begin{algo}{T}{Track runs}
  The algorithm takes a $\varphi \in \LTLFO$ (under a corresponding
  valuation $v$), for which it creates an SA, $\cA_{\varphi,v}$.  It
  then reads an event $(\struct, \sigma)$ and returns, if
  $\cA_{\varphi,v}$, after processing $(\struct, \sigma)$, does not have
  any potentially locally accepting runs, for which
  obligations haven't been detected as violated. Otherwise, it saves the new state
  of $\cA_{\varphi,v}$, waits for new input, and then checks again, and
  so forth.
  \setlist{nosep}
  \begin{description}[style=multiline,leftmargin=0.8cm]
  \item[\,\,\,T1.] [Create SA.] Create an SA, $\cA_{\varphi,v}$, in the
    usual manner.
  \item[\,\,\,T2.\namedlabel{itm:read}{T2}] [Wait for new event.] Let
    $(\struct, \sigma)$ be the event that was read.
  \item[\,\,\,T3.\namedlabel{algo:r:buffer}{T3}] [Buffer with runs.]
    Let $B$ and $B'$ be (initially empty) buffers.  If $B =
    \emptyset$, for each $q \in Q_0$ and for each $q' \in \ddelta(q,
    (\struct, \sigma))$: add $(q', [\spawn(q, (\struct, \sigma))])$ to
    $B$.
    Otherwise, set $B' = B$, and subsequently $B = \emptyset$. Next,
    for all $(q, [obl_1, \ldots, obl_n])\allowbreak \in B'$ and for
    all $q' \in \ddelta(q, (\struct,\sigma))$: add
    $(q',[obl_{\textsl{new}}, obl_1,\ldots, \allowbreak obl_n])$ to
    $B$, where $obl_{\textsl{new}}=\spawn(q, (\struct, \sigma))$.
  \item[\,\,\,T4.\namedlabel{itm:init_submonitors}{T4}] [Create
    submonitors.]  For each $(q, [obl_{\textsl{new}}, obl_1 \ldots,
    obl_n]) \in B$: call Algorithm~M with argument $\psi$ (under
    corresponding valuation $v'$) for each $\cA_{\psi,v'}$ that occurs
    in $obl_{\textsl{new}}$.
  \item[\,\,\,T5.\namedlabel{itm:iterate_buf}{T5}] [Iterate over
    candidate runs.]  Assume $B =\{ b_0, \ldots, b_m \}$. Create a
    counter $j = 0$ and set $(q, [obl_0, \ldots, obl_n]) = b_j$ to be
    the $j$th element of $B$.
  \item[\,\,\,T6.\namedlabel{itm:send}{T6}] [Send, receive, replace.]
    For all $0 \leq i \leq n$: send $(\struct, \sigma)$ to all
    submonitors corresponding to SAs occurring in $obl_i$, and wait
    for the respective verdicts.  For every returned $\top$ (resp.\
    $\bot$) replace the corresponding SA in $obl_i$ with $\top$
    (resp.\ $\bot$).
  \item[\,\,\,T7.\namedlabel{itm:rm_seq}{T7}] [Corresponding run has
    violated obligations?]
    For all $0 \leq i \leq n$: if $obl_i = \bot$, remove $b_j$ from
    $B$, set $j$ to $j + 1$, and go to \ref{itm:send}.
  \item[\,\,\,T8.\namedlabel{algo:r:rm_obl}{T8}] [Obligations met?]
    For all $0 \leq i \leq n$: if $obl_i = \top$, remove
    $obl_i$. 
  \item[\,\,\,T9.\namedlabel{itm:next_buf}{T9}] [Next run in buffer.]
    If $j \leq m$, set $j$ to $j+1$ and repeat step \ref{itm:send}.
  \item[T10.\namedlabel{itm:result}{T10}] [Communicate verdict.]  If $B
    = \emptyset$, send ``no runs'' to the calling Algorithm~M and return,
    otherwise send ``some run(s)'' and go back to \ref{itm:read}.
    %
    \hspace{1em}\Rectangle
  \end{description}
  \vspace{1em}
\end{algo}

\noindent
For a given $\varphi \in \LTLFO$ and $(\tstruct, u)$, let us use 
$M_\varphi(\tstruct, u)$ to denote the successive application of
Algorithm~M for formula $\varphi$ first on $u_0$, then $u_1$, and so
forth.  We then get

\begin{theorem} 
  \label{thm:mon:correct}
  $M_{\varphi}(\tstruct, u) = \top \Rightarrow$ $(\tstruct, u) \in
  \good(\varphi)$ (resp.\ for $\bot$ and $\bad(\varphi)$).
\end{theorem}
\begin{proof}[Idea]
  By nested induction over $\depth(\varphi)$ and the length of
  $(\ftstruct, u)$.
\end{proof}


\subsection{Experimental results}

To demonstrate the feasibility of our proposed algorithm and to
get an intuition on its runtime performance (i.e., average space
consumption at runtime), we have implemented the above.
The only liberty we took in deviating from our description is the
following: 
since the SAs for $\varphi \in \LTLFO$ on the different levels basically
consist of ordinary BAs 
for the respective subformulae of $\varphi$, we have used an ``off the
shelf'' BA 
generator,
lbt\footnote{http://www.tcs.hut.fi/Software/maria/tools/lbt/}, instead
of expanding the state-space ourselves.
%
We also compared our implementation with the somewhat naive (but,
arguably, 
easier to implement) approach of monitoring $\LTLFO$ formulae, described
in \cite{bauer:kuester:vegliach:NFM12}. 
%
There,
we 
used the well-known concept of formula rewriting, sometimes referred to
as progression:
a function, $P$, continuously ``rewrites'' a formula $\varphi \in
\LTLFO$ using an observed event, $\sigma$, in order to obtain a new
formula, $\varphi'$, that states what has to be true now and what in the
future.  If $\varphi' = \top$, then $\sigma \in \good(\varphi')$, if
$\varphi' = \bot$ then $\sigma \in \bad(\varphi')$, otherwise the
thereby realised monitor waits for further events to apply its
progression function to.  $P$ rewrites according to the well-known
fixpoint characterisations of $\LTL$ operators, such as $P(\ltlG\varphi,
\sigma) = P(\varphi,\sigma) \wedge \ltlG\varphi$.  This is a well
established principle to evaluate LTL formulae over traces in a stepwise
manner (cf.\ \cite{Bacchus:1998:PTE:590220.590230}).
%

Some results of this comparison are visualised in \figref{fig:results}.
For each $\LTLFO$ formula, we randomly generated $5$
traces 
of length $100$ and passed them to the respective algorithm.  The
$x$-axis marks the trace length, the $y$-axis the space consumption of
the monitors; that is, the length of the formula after progression vs.\
the number of automata states.
In graph (a) the divergence between both approaches is the most striking
as it highlights one of the potential problems of progression, namely
that a lot of redundant information can accumulate: If $\forall x:p.\
r(x)$ ever becomes true, then $P$ will produce a new conjunct $\ltlG
\forall y:q.\ s(y)$ for each new event, even though semantically it
makes no difference.
%
%
%
In comparison, the automata-based monitor's size, measured in terms of
the number of SA states, stays more or less constant throughout the
trace.  This can be explained by the fact that for syntactically
different, but semantically equivalent formulae our
BA 
generator usually produces the same automaton (as is clearly the case in
this example).
%

With minor but noteworthy exceptions, the straight blue lines of (b)
and (d) mirror (and scale) the dashed black lines, which means that
our monitor is on average smaller by some degree, but in the long run
not substantially smaller.  Note how, unlike in (b), the straight blue
lines in (d) are not \emph{exact} scaled copies of the black dashed
lines, in that the graph depicting the performance of progression has
a number of spikes.  As the input traces 
for the monitors are randomly generated, the time when $\forall y: q.\
\ltlX s(y)$ becomes true differs, and hence the size of the progressed
formula may increase, whereas the automata-based monitor stays small
for the same reasons as outlined above in (a).

Finally, the graph of (c) is interesting in that both monitors show a
tendency to grow over time.  The reason for that is that the right hand
side of the $\ltlU$-operator in (c), $\ltlX s(x)$, makes use of the same
$x$ which is quantified on the left hand side.  For example, if the
events are given by $\{ \{ p(1), p(2) \}, \{ p(3), p(4), p(5) \} \}$,
the monitor would have to remember all the domain elements of the
$\preds$-operator until $\ltlX s(1) \wedge \ltlX s(2)$ and $\ltlX s(3)
\wedge \ltlX s(4) \wedge \ltlX s(5)$ hold.
Depending on how late in the trace 
this is the case (if ever), memory consumption increases for both
monitors.

\begin{figure}[htb]
  \centering
  \scalebox{.75}{%
  \hspace{-2em}
  \subfigure[$\ltlG (\forall x:p.\ r(x) \Rightarrow \ltlG \forall y:q.s(y))$]{
    \includegraphics[scale=0.3] {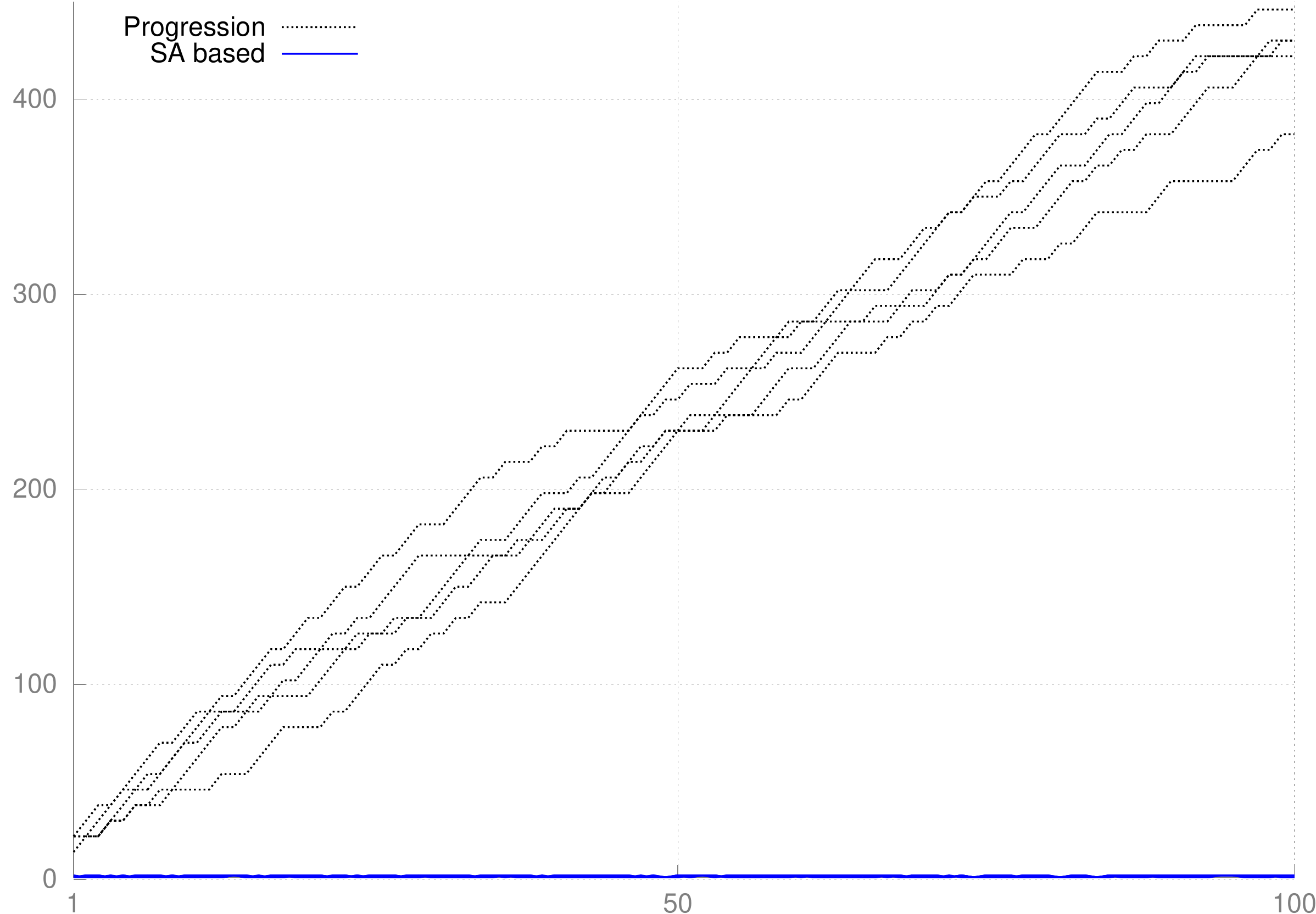}
    \label{fig:subfig1}
  }%
  \subfigure[$\ltlG (\forall x:p.\ r(x) \Rightarrow \ltlX s(x))$]{
    \includegraphics[scale=0.3] {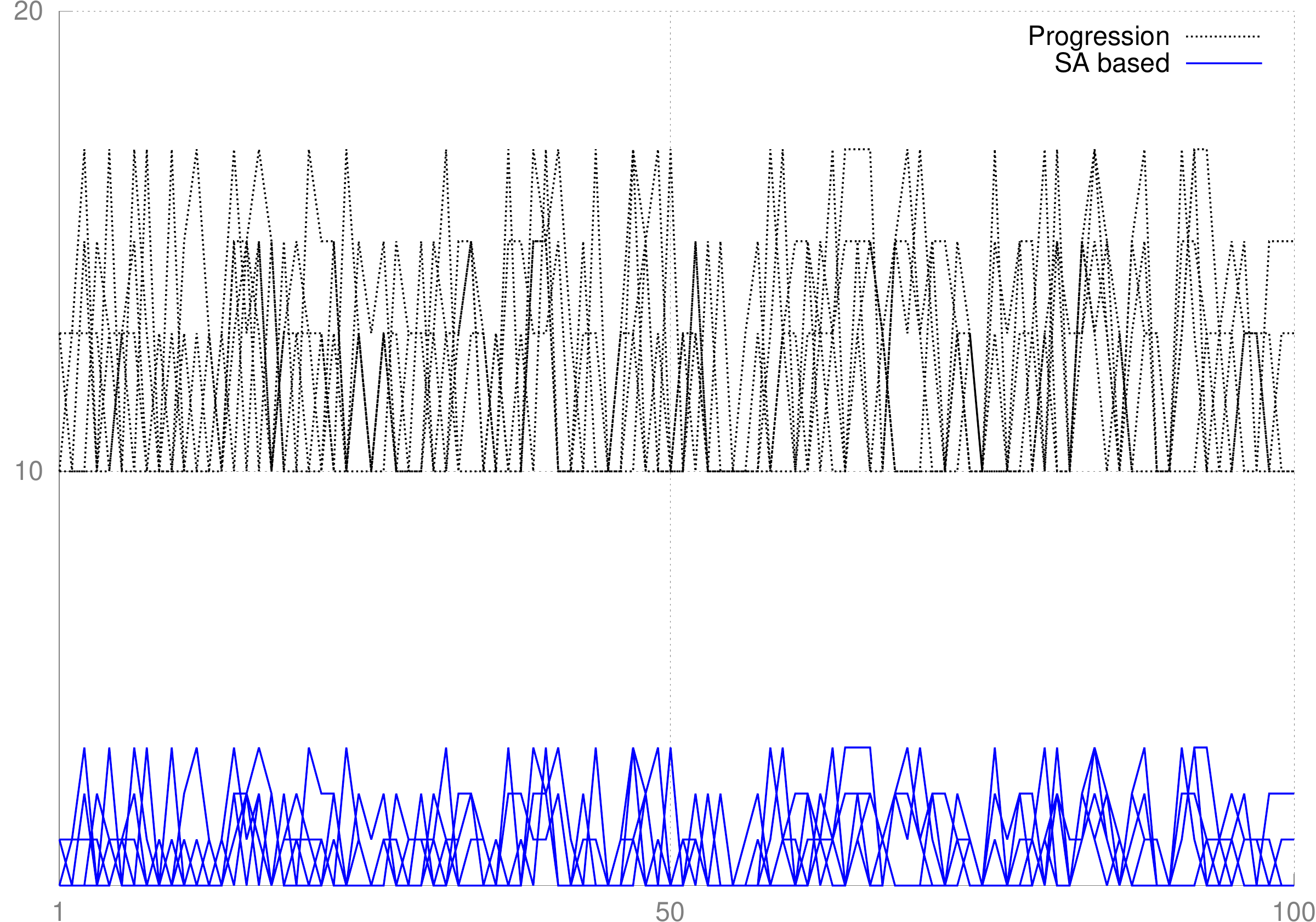}
    \label{fig:subfig2}
  }}
  \scalebox{.75}{%
  \hspace{-2em}
  \subfigure[$\ltlG (\forall x:p.\ r(x) \ltlU \ltlX s(x))$]{
    \includegraphics[scale=0.3] {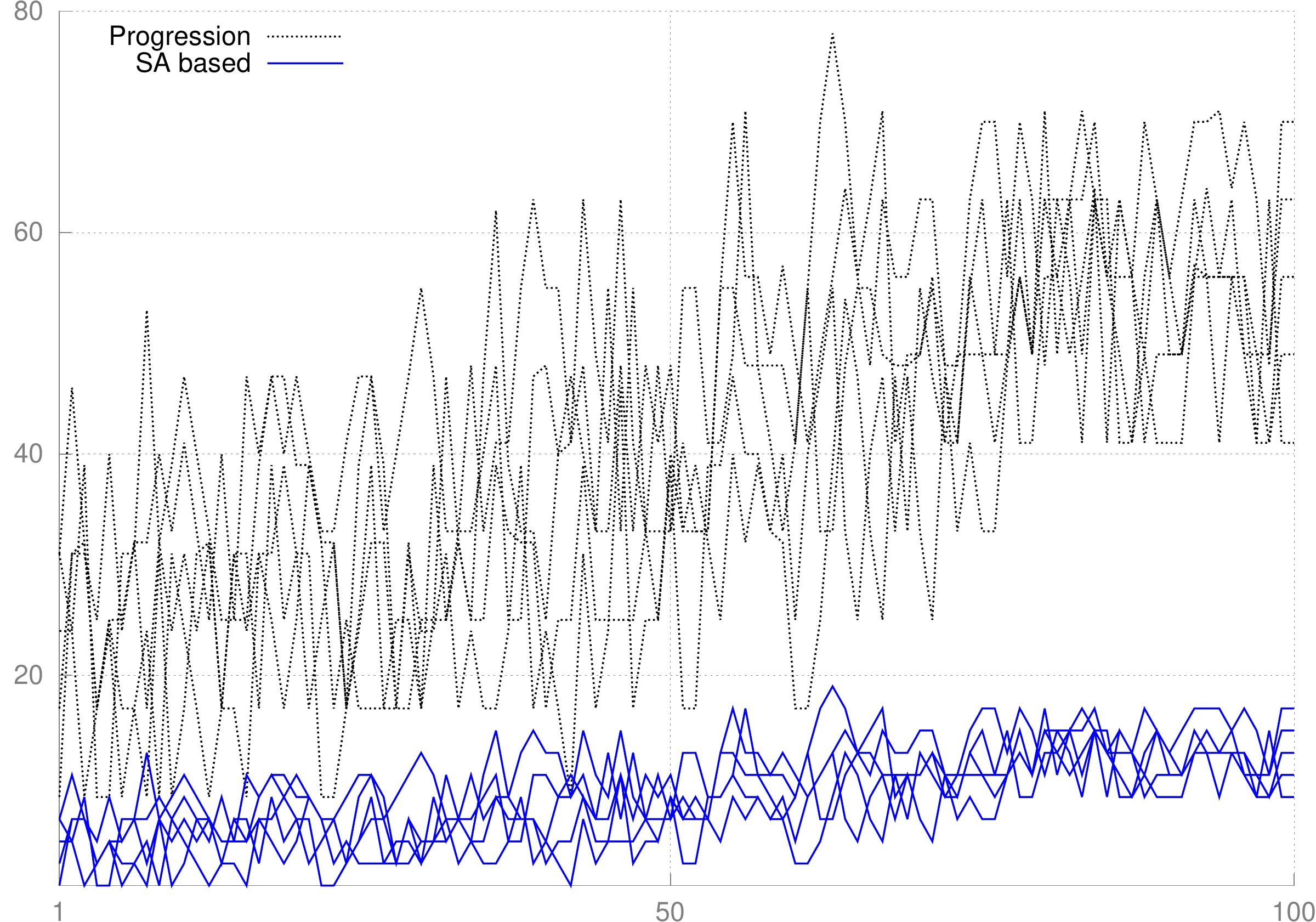}
    \label{fig:subfig3}
  }%
  \subfigure[$\ltlG (\forall x:p.\ r(x) \ltlU_w \forall y:q.\ \ltlX s(y))$]{
    \includegraphics[scale=0.3] {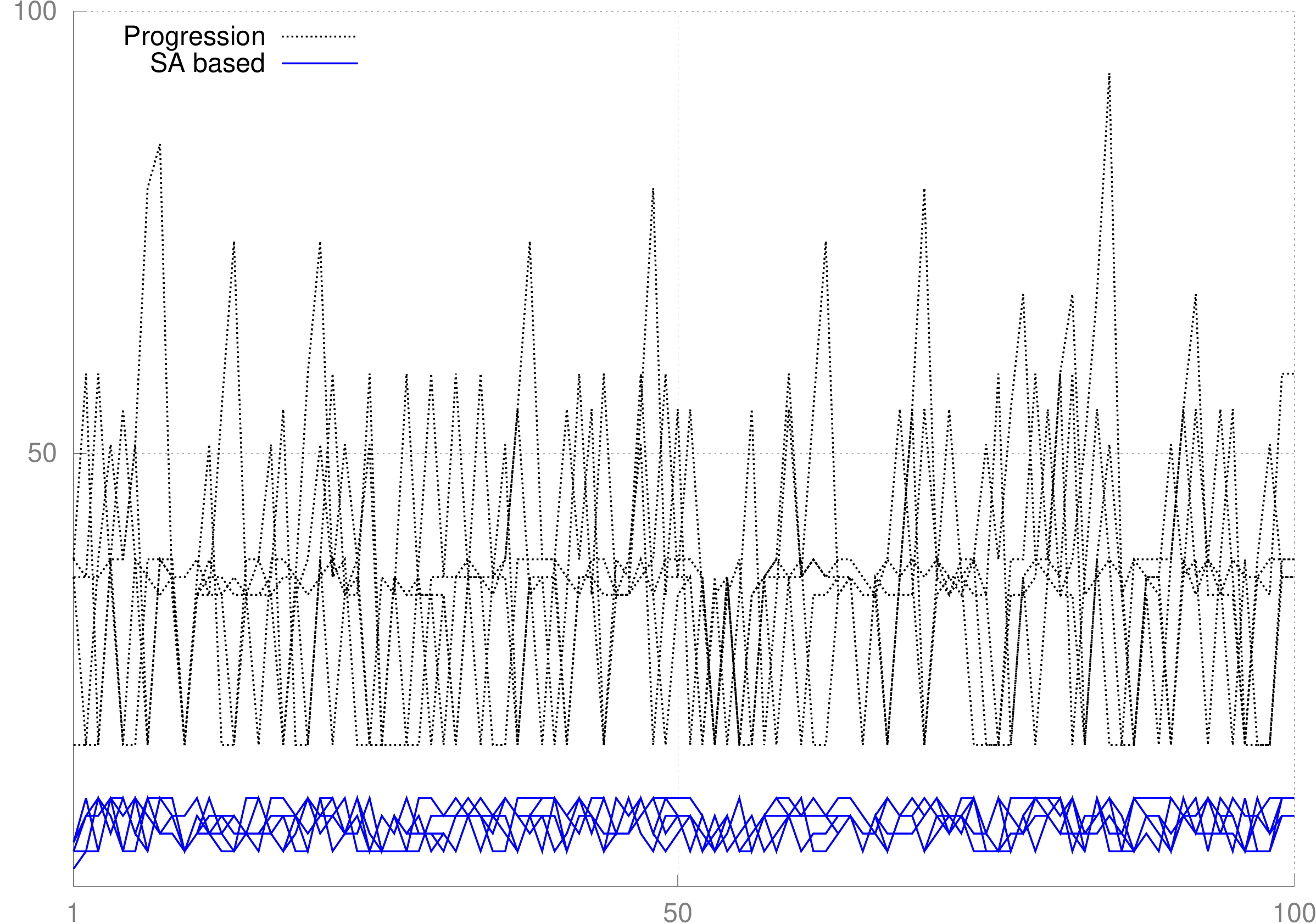}
    \label{fig:subfig3}
  }}
  \caption{Automata- (blue lines) vs.\ progression-based (black dashed lines) $\LTLFO$-monitoring.}
  \label{fig:results}
\end{figure}




\section{Related work}
\label{sec:related}
This is by no means the first work to discuss monitoring of first-order
specifications.
Mainly motivated by checking temporal triggers and temporal constraints,
the monitoring problem for different types of first-order logic has been
widely studied in the database community, for example.
In that context, Chomicki \cite{DBLP:journals/tods/Chomicki95} presents
a method to check for violations of temporal constraints, specified
using (metric) past temporal operators. The logic in
\cite{DBLP:journals/tods/Chomicki95} differs from $\LTLFO$ in that it
allows natural first-order quantification over a single countable and
constant domain, whereas quantified variables in $\LTLFO$ range over
elements that occur at the current position of the trace (see also
\cite{Halle:2008:RMM:1437901.1438836,bauer:gore:tiu:ictac09}).
Presumably, to achieve the same effect,
\cite{DBLP:journals/tods/Chomicki95} demands that policies are what is
called ``domain independent'', so that statements are only ever made
\wrt known objects.  As such, domain independence is a property of the
policy and shown to be undecidable.  In contrast, one could say that
$\LTLFO$ has a similar notion of domain independence already built-in,
because of its quantifier.
Like $\LTLFO$, the logic in \cite{DBLP:journals/tods/Chomicki95} is also
undecidable; no function symbols are allowed and relations are required
to be finite.
%
%
However, despite the fact that the prefix problem is not phrased as a
decision problem, its basic idea is already denoted by Chomicki under
the notion of a potential constraint satisfaction problem.  In
particular, he shows that the set of prefixes of models for a given
formula is not recursively enumerable.  On the other hand, the monitor
in \cite{DBLP:journals/tods/Chomicki95} does not tackle this problem and
instead solves what we have introduced as the word problem, which,
unlike the prefix problem, is decidable.

Basin et al.\ \cite{Basin_etal:monitoring_mfotl} extend Chomicki's
monitor towards bounded future operators using the same logic.
Furthermore, they allow infinite relations as long as these are
representable by automatic structures, i.e., automata models.  In this
way, they show that the restriction on formulae to be domain independent
is no longer necessary.  $\LTLFO$, in comparison, is more general,
in that it allows computable relations and functions.

The already cited work of Hall\'e and Villemaire
\cite{Halle:2008:RMM:1437901.1438836} describes a monitor for a logic
with quantification identical to ours, but without function symbols
and only equality instead of arbitrary computable
relations. Furthermore, the size of the individual worlds is a priori
bounded by a fixed value. Additionally, their monitor is fully
generated ``on the fly'' by using syntax-based decomposition rules,
similar to formula progression.  In our approach, however, it is
possible to pre-compute the individual BAs for the respective
subformulae of a policy/levels of the SA, and thereby bound the
complexity of that part of our monitor at runtime by a constant
factor.

Sistla and Wolfson \cite{DBLP:journals/tkde/SistlaW95} also discuss a
monitor for database triggers whose conditions are specified in a logic,
which uses an assignment quantifier that binds a single value or a
relation instance to a global, rigid variable. Their monitor is
represented by a graph structure, which is extended by one level for
each updated database state, and as such proportional in size to the
number of updates.


\section{Conclusions}
\label{sec:conc}

To the best of our knowledge, our monitoring algorithm is the first to
devise anticipatory monitors, i.e., address the prefix problem instead
of a (variant of the) word problem, for policies given in an undecidable
first-order temporal logic.  Moreover, unlike other approaches, such as
\cite{DBLP:journals/tkde/SistlaW95,Halle:2008:RMM:1437901.1438836} and
even \cite{bauer:kuester:vegliach:NFM12}, we are able to precompute most
of the state space required at runtime (i.e., replace step T1 in
Algorithm~T with a look-up in a precomputed table of SAs and merely use
a new valuation), as the different levels of our SAs correspond to more
or less standard BAs that can be generated before monitoring commences.
Moreover as required, our monitor is monotonic and in principle
trace-length independent.  The latter, however, deserves closer
examination.
Consider the formula given in \figref{fig:results}~(c): it basically
forces the monitor to memorise all occurrences of $p$ 
in every event and keep them until $s(x)$ holds, respectively.  If
$s(x)$ never holds (or not for a very long time), the space consumption
of the monitor is bound to grow.  Hence, unlike in standard $\LTL$,
trace-length dependence is not merely a property of the monitor, but
also of the specification.  We have not yet investigated whether
trace-length dependence is decidable and if so, at what cost.  However,
if the formula is \emph{not} trace-length dependent, then our monitor is
trace-length independent, as desired.
Given a $\varphi \in \LTLFO$ of which we know that it is trace-length
independent in principle, our monitor's size at runtime at any given
time is bounded by $O(|\sigma|^{\depth(\varphi)} \cdot
2^{|cl(\varphi)|})$, where $\sigma$ is the current input to the monitor:
Throughout the $\depth(\varphi)$ levels of the monitor, there are a
total of $O(|\sigma|^{\depth(\varphi)})$ ``submonitors'', which are of
size $O(2^{|cl(\varphi)|})$, respectively.
In contrast, the size of a progression-based monitor, even for
obviously trace-length independent formulae, such as given in
\figref{fig:results}~(a) is, in the worst case, proportional to the
length of the trace so far.

In Table~1 we have summarised the main results of
\S\ref{sec:prop_case}--\S\ref{sec:fo_case}, highlighting again the
differences of $\LTL$ compared to $\LTLFO$.  Note that as far as
trace-length dependence goes, for $\LTL$ it is always possible to devise
a trace-length independent monitor (cf.~\cite{bauer:leucker:schallhart:tosem}).

\begin{table}[tb]
  \centering
  \caption{Overview of complexity results.}
  \label{tab:res}
  \def\arraystretch{1.3}%
  \scalebox{0.9}{%
  \begin{tabular}{|l|c|c|p{4cm}|c|}
    \hline
    & \textbf{Satisfiability} & \textbf{Word problem} & \centering \textbf{Model checking} & \textbf{Prefix problem}\\
    \hline
    $\LTL$ & PSpace-complete & $<$ Bilinear-time & \centering PSpace-complete & PSpace-complete \\
    \hline
    $\LTLFO$ & Undecidable & PSpace-complete & \centering ExpSpace-membership, PSpace-hard & Undecidable \\
    \hline
  \end{tabular}}
\end{table}

\paragraph{Acknowledgements.}
Our thanks go to Patrik Haslum, Michael Norrish and Peter Baumgartner
for helpful comments on earlier drafts of this paper.



\bibliographystyle{abbrv}
\bibliography{shortstrings,bibliography}

\begin{thebibliography}{10}

\bibitem{Bacchus:1998:PTE:590220.590230}
F.~Bacchus and F.~Kabanza.
\newblock Planning for temporally extended goals.
\newblock {\em Annals of Mathematics and Artificial Intelligence}, 22:5--27,
  1998.

\bibitem{Baier:2008:PMC:1373322}
C.~Baier and J.-P. Katoen.
\newblock {\em Principles of Model Checking}.
\newblock MIT Press, 2008.

\bibitem{Basin_etal:monitoring_mfotl}
D.~Basin, F.~Klaedtke, and S.~M\"uller.
\newblock Policy monitoring in first-order temporal logic.
\newblock In {\em Proc.\ 22nd Intl.\ Conf.\ on Computer Aided Verification
  (CAV)}, volume 6174 of {\em LNCS}, pages 1--18. Springer, 2010.

\bibitem{bauer:gore:tiu:ictac09}
A.~Bauer, R.~Gore, and A.~Tiu.
\newblock A first-order policy language for history-based transaction
  monitoring.
\newblock In {\em Proc.\ 6th Intl.\ Colloq.\ on Theoretical Aspects of
  Computing (ICTAC)}, volume 5684 of {\em LNCS}, pages 96--111. Springer, 2009.

\bibitem{bauer:kuester:vegliach:NFM12}
A.~Bauer, J.-C. K\"uster, and G.~Vegliach.
\newblock Runtime verification meets {A}ndroid security.
\newblock In {\em Proc.\ 4th NASA Formal Methods Symp. (NFM)}, volume 7226 of
  {\em LNCS}, pages 174--180. Springer, 2012.

\bibitem{bauer:leucker:schallhart:tosem}
A.~Bauer, M.~Leucker, and C.~Schallhart.
\newblock Runtime verification for {LTL} and {TLTL}.
\newblock {\em ACM Transactions on Software Engineering and Methodology},
  20(4):14, 2011.

\bibitem{DBLP:journals/tods/Chomicki95}
J.~Chomicki.
\newblock Efficient checking of temporal integrity constraints using bounded
  history encoding.
\newblock {\em ACM Trans. Database Syst.}, 20(2):149--186, 1995.

\bibitem{DBLP:journals/jcss/ChomickiN95}
J.~Chomicki and D.~Niwinski.
\newblock On the feasibility of checking temporal integrity constraints.
\newblock {\em J. Comput. Syst. Sci.}, 51(3):523--535, 1995.

\bibitem{Dong:2008:IAR:1478278.1478320}
W.~Dong, M.~Leucker, and C.~Schallhart.
\newblock Impartial anticipation in runtime-verification.
\newblock In {\em Proc.\ 6th Intl.\ Symposium on Automated Technology for
  Verification and Analysis (ATVA)}, volume 5311 of {\em LNCS}, pages 386--396.
  Springer, 2008.

\bibitem{DBLP:conf/cav/EisnerFHLMC03}
C.~Eisner, D.~Fisman, J.~Havlicek, Y.~Lustig, A.~McIsaac, and D.~V. Campenhout.
\newblock Reasoning with temporal logic on truncated paths.
\newblock In {\em Proc.\ 15th Intl.\ Conf.\ on Computer Aided Verification
  (CAV)}, volume 2725 of {\em LNCS}, pages 27--39. Springer, 2003.

\bibitem{578533}
M.~R. Garey and D.~S. Johnson.
\newblock {\em Computers and Intractability: A Guide to the Theory of
  NP-Completeness}.
\newblock W. H. Freeman \& Co., New York, NY, USA, 1979.

\bibitem{DBLP:conf/fm/GenonMM06}
A.~Genon, T.~Massart, and C.~Meuter.
\newblock Monitoring distributed controllers: When an efficient {LTL} algorithm
  on sequences is needed to model-check traces.
\newblock In {\em Proc.\ 14th Intl.\ Symp.\ on Formal Methods (FM)}, volume
  4085 of {\em LNCS}, pages 557--572. Springer, 2006.

\bibitem{Halle:2008:RMM:1437901.1438836}
S.~Halle and R.~Villemaire.
\newblock Runtime monitoring of message-based workflows with data.
\newblock In {\em Proc. 12th IEEE Enterprise Distributed Object Computing
  Conference (EDOC)}, pages 63--72. IEEE, 2008.

\bibitem{havelund2004}
K.~Havelund and G.~Rosu.
\newblock Efficient monitoring of safety properties.
\newblock {\em Software Tools for Technology Transfer}, 6(2):158--173, 2004.

\bibitem{DBLP:journals/corr/abs-1210-0574}
L.~Kuhtz and B.~Finkbeiner.
\newblock Efficient parallel path checking for linear-time temporal logic with
  past and bounds.
\newblock {\em Logical Methods in Computer Science}, 8(4), 2012.

\bibitem{Libkin:2004:EFM:1024196}
L.~Libkin.
\newblock {\em Elements Of Finite Model Theory}.
\newblock Springer, 2004.

\bibitem{DBLP:conf/concur/MarkeyS03}
N.~Markey and P.~Schnoebelen.
\newblock Model checking a path.
\newblock In {\em Proc.\ 14th Int.\ Conf.\ on Concurrency Theory (CONCUR)},
  volume 2761 of {\em LNCS}, pages 248--262. Springer, 2003.

\bibitem{DBLP:journals/jacm/SistlaC85}
A.~P. Sistla and E.~M. Clarke.
\newblock The complexity of propositional linear temporal logics.
\newblock {\em J. ACM}, 32(3):733--749, 1985.

\bibitem{DBLP:journals/tkde/SistlaW95}
A.~P. Sistla and O.~Wolfson.
\newblock Temporal triggers in active databases.
\newblock {\em IEEE Trans. Knowl. Data Eng.}, 7(3):471--486, 1995.

\end{thebibliography}


\clearpage
\appendix
\section{Detailed proofs}
\noindent\textbf{Lemma~\ref{lem:finmodel}.  }
Let $\varphi$ be a sentence in first-order logic, then we can
construct a corresponding $\psi \in \LTLFO$ \st $\varphi$ has a finite
model iff $\psi$ is satisfiable.

\begin{proof}
  We construct $\psi$ as follows. 
  We first introduce a new unary $\preds$-operator $d$ whose arity is
  $\tau$ and that does not appear in $\varphi$.  We then replace every
  subformula in $\varphi$, which is of the form $\forall x.\ \theta$,
  with $\forall x:d.\ \theta$ (resp.\ for $\exists x.\ \theta$).
  Next, we encode some restrictions on the interpretation of function
  and predicate symbols:
  \begin{itemize}
  \item For each constant symbol $c$ in $\varphi$, we conjoin the
    obtained $\psi$ with $d(c)$.
  \item For each function symbol $f$ in $\varphi$ of arity $n$, we
    conjoin the obtained $\psi$ with
    %
    $\forall x_1:d.\ \ldots \forall x_n:d.\ d(f(x_1, \ldots, x_n))$.
  \item For each predicate symbol $p$ in $\varphi$ of arity $n$, we
    conjoin the obtained $\psi$ with $\forall(x_1, \ldots, x_n):p.\
    d(x_1) \wedge \ldots \wedge d(x_n)$.
  \item We conjoin $\exists x:d.\ d(x)$ to the obtained $\psi$ to ensure
    that the domain is not empty.
  \end{itemize}
  Finally, we fix the arities of symbols in $\psi$ appropriately to one
  of the following $\tau$, $\tau \times \ldots \times \tau$, $\tau
  \times \ldots \times \tau \rightarrow \tau$.

  Obviously, the formula $\psi$, constructed by the procedure above, is
  a syntactically correct $\LTLFO$ formula.
  Now, if $\psi$ is satisfiable by some $(\struct', \sigma)$, where
  $\struct' = (|\struct'|, I')$ and $\sigma \in (\struct')$-$\events$, it is easy to
  construct a finite model $\struct = (\dom, I)$ \st $\struct \models
  \varphi$ holds in the classical sense of first-order logic: set $\dom
  = d^{I'}$, $c^I = c^{I'}$, $f^I = f^{I'}|_{d^{I'} \times \ldots \times
    d^{I'}}$, $p^I=p^{I'}$, 
  respectively.
  %
  %
  By an inductive argument one can show that the \LTLFO semantics is
  preserved.  The other direction, if $\varphi$ is finitely satisfiable,
  is trivial: set $|\struct'| = \tau^{I'} = \dom$, $c^{I'} = c^I$,
  $f^{I'} = f^I$, respectively, and $\sigma = \{ (p,\vec{e}) \mid \vec{e} \in p^I\} \cup \{ (d,e) \mid e \in |\struct|\}$.
  %
  \qed
\end{proof}

\noindent\textbf{Theorem~\ref{lem:fo:word}.  }
The word problem for $\LTLFO$ is PSpace-complete.

\begin{proof}
  To evaluate a formula $\varphi \in \LTLFO$ over some linear Kripke
  structure, $\cK$, we can basically use the inductive definition of the
  semantics of \LTLFO: If used as a function, starting in the initial
  state of $\cK$, $s_0$, it evaluates $\varphi$ in a depth-first manner
  with the maximal depth bounded by $|\varphi|$.

  To show hardness, we reduce the following problem, which is known to
  be PSpace-complete: Let $F = Q_1x_1.\ Q_2x_2.\ \ldots Q_nx_n.\ E(x_1,
  x_2, \ldots, x_n)$, where $Q \in \{ \forall, \exists \}$ and $E$ is a
  Boolean expression over variables $x_1, x_2, \ldots, x_n$. Does $F$
  evaluate to $\top$ (cf.\ \cite{578533})?  The reduction of this
  problem proceeds as follows.  We first construct a formula $\varphi
  \in \LTLFO$ in prenex normal form,
  \[
  \varphi = Q_1x_1:d.\ Q_2x_2:d.\ \ldots Q_nx_n:d.\ E(p_{x_1}(x_1),
  p_{x_2}(x_2), \ldots, p_{x_n}(x_n)).
  \]
  Then, using an $\preds$-operator $p_{x_i}$ for every variable $x_i$,
  we construct a singleton Kripke structure, $\cK$, \st $\lambda(s_0) =
  (\struct, \{ (d,0), (d,1), (p_{x_1},1), (p_{x_2},1), \ldots,
  (p_{x_n},1) \}), $ where $\dom = \{ 0, 1 \}$ and $I$ defined
  accordingly.
  It can easily be seen that $F$ evaluates to $\top$ iff $\cK$ is a
  model for $\varphi$. Moreover, this construction can be obtained in no
  more than a polynomial number of steps \wrt the size of the
  input. \qed
\end{proof}

\noindent\textbf{Theorem~\ref{lem:fo:mc}.  }
The model checking problem for $\LTLFO$ is in ExpSpace.

\begin{proof}
  For a given $\varphi \in \LTLFO$ and $(\struct)$-Kripke structure
  $\cK$ defined as usual, where $\struct = (\dom, I)$, we construct a
  propositional Kripke structure $\cK'$ and $\varphi' \in \LTL$, \st
  $\cL(\cK) \subseteq \cL(\varphi)$ iff $\cL(\cK') \subseteq
  \cL(\varphi')$ holds.
  Assuming variable names in $\varphi$ have been adjusted so that each
  has a unique name, the construction of $\varphi'$ proceeds as follows.

  Wlog.\ we can assume $\dom$ to be a finite set $\{ d_0, \ldots, d_n
  \}$.
  We first set $\varphi'$ to $\varphi$ and extend the corresponding
  $\Gamma$ by the constant symbols $c_{d_0}, \ldots, c_{d_n}$, \st
  $c_{d_i}^I = d_i$, respectively;
  that is, we add the respective interpretations of each $c_{d_i}$ to
  $I$.  This step obviously does not require more than polynomial space.
  We then replace all subformulae in $\varphi'$ of the form $\nu =
  Q\vx:p.\ \psi(\vx)$ exhaustively with the following constructed
  $\psi'$:

  \begin{itemize}
  \item Set $\psi' = \top$.
  \item For each state $s \in S$ do the following:
    \begin{itemize}
    \item Let $T = \{ \vd \mid \lambda(s) = (\struct', \sigma),
      \struct' \sim \struct \hbox{ and } (p,\vd) \in \sigma \}$.
    \item If $Q = \forall$, then
      $$
      \psi' = \psi' \wedge (\tilde{s} \Rightarrow \bigwedge_{\vd \in T}
      \psi(\vx)[\vc/\vx]), \hbox{where } \vc \hbox{ is \st\ } \vc^I = \vd,
      $$ 
      otherwise
      $$
      \psi' = \psi' \wedge (\tilde{s} \Rightarrow \bigvee_{\vd \in T}
      \psi(\vx)[\vc/\vx]), \hbox{where } \vc \hbox{ is \st\ } \vc^I = \vd,
      $$
      where $\tilde{s}$ is a fresh, unique predicate symbol meant
      to represent state $s$.
    \end{itemize}
  \end{itemize}
  Then, for all subformulae in $\varphi'$ of the form $\tilde{s}
  \Rightarrow \psi$ we do the following:
  \begin{itemize}
  \item For each $r(\vt)$ occurring in $\psi$, where $r \in \Rels$ and
    $\vt$ are terms, let $\vd = \vt^I$, and replace $r(\vt)$ by a fresh,
    unique predicate symbol $r_{\vd}$.
  \end{itemize}
  It is easy to see that, indeed, $\varphi'$ is a syntactically correct
  standard LTL formula, where all quantifiers have been eliminated.  In
  terms of space complexity, note that in the first loop, we replace
  each quantified formula by an 
  expression at least $|\cK|$ times longer than the original quantified
  formula.  In the worst case, the final formula's length will be
  exponential in the number of quantifiers.

  We now define the propositional Kripke structure $\cK' = (S', s'_0,
  \lambda', \rightarrow')$ as follows.  Let $S' = S$, $s'_0 = s_0$, and
  $\rightarrow' = \rightarrow$.  In what follows, let $s$ be a state and
  $\lambda(s) = ((|\struct|, I), \sigma)$.  (Note, this is the labelling
  function of $\cK$.)  The alphabet of $\cK'$ is given by $2^{\AP}$,
  where $\AP = \{ r_{\vd} \mid r \in \Rels \hbox{ and } \vd \in \dom
  \} \cup \{\tilde{s} \mid s \in S\}$.
  Finally, we define the labelling function of $\cK'$ as
  $
  \lambda'(s) = \{ \tilde{s} \} \cup \{ r_{\vd} \mid r \in \Rels \hbox{ and } r^I(\vd) \hbox{ is true} \}.
  $
  It is easy to see that, indeed, $\cK'$ preserves all the runs possible
  through $\cK$.

  One can show by an easy induction on the structure of $\varphi'$ that,
  indeed, $\cL(\cK) \subseteq \cL(\varphi)$ iff $\cL(\cK') \subseteq
  \cL(\varphi')$ holds.
  \qed
\end{proof}

\noindent\textbf{Lemma~\ref{lem:fo:restricted}.  }
Let
$\struct$ be a first-order
structure and $\varphi \in \LTLFO$, then
$\cL(\varphi)_{\struct} =
\{ (\tstruct, w) \mid \tstruct \sim \struct, w \in (2^{\events})^\omega,
\hbox{ and } (\tstruct, w) \models \varphi \}$.
%
Testing if $\cL(\varphi)_{\struct} \neq \emptyset$ is generally
undecidable.

\begin{proof}
  Let $K = (x_1, y_1), \ldots, (x_k, y_k)$ be an instance of Post's
  Correspondence Problem over $\Sigma = \{ 0, 1 \}$, where $x_i,y_i
  \in \Sigma^+$, which is known to be undecidable in this
  form.
  Let us now define a formula $\varphi_K = \exists \gamma:z.\
  pcp(\gamma)$, a structure $\struct = (\Sigma^+, I)$, \st
  $pcp^I(u) \Leftrightarrow u = x_{i_1}\ldots x_{i_n} = y_{i_1}\ldots
  y_{i_n}$, where $u \in \Sigma^+$ 
  and $pcp$ is of corresponding arity.
  Obviously, $pcp^I(u)$ can be computed in finite time for any given
  $u$.
  Let us now show that $\cL(\varphi_K)_\struct \neq \emptyset$ iff $K$ has
  a solution.
  %

  ($\Rightarrow$:) Because $\mathcal{L}(\varphi_K)_\struct \neq
  \emptyset$, let's assume there is a word $u \in \Sigma^+$ st.\ $(z,u)
  \in \sigma$ and \ $(\struct,\sigma) \in
  \mathcal{L}(\varphi_K)_{\struct}$.  By the choice of $pcp^I$,
  there exists a sequence of indices, $i_1, \ldots, i_n$, st.\ $u =
  x_{i_1}\ldots x_{i_n} = y_{i_1}\ldots y_{i_n}$, i.e., $K$ has a
  solution.

  ($\Leftarrow$:)
  Let's assume $K$ has a solution, i.e., there exists a word $u \in
  \Sigma^+$ and a sequence of indices, $i_1, \ldots, i_n$, st.\ $u =
  x_{i_1}\ldots x_{i_n} = y_{i_1}\ldots y_{i_n}$. We now have to show
  that $\mathcal{L}(\varphi_K)_{\struct} \neq \emptyset$. For this
  purpose, set $\sigma = \{ (z,u) \}$, then $(\struct,\sigma) \in
  \mathcal{L}(\varphi_K)_{\struct}$ and, consequently,
  $\mathcal{L}(\varphi_K)_{\struct} \neq \emptyset$.
  \qed
\end{proof}

\noindent\textbf{Theorem~\ref{lem:fo:pf}.  }
The prefix problem for $\LTLFO$ is undecidable.
\begin{proof}
  By way of a similar reduction used in Theorem~\ref{thm:prop:prefix}
  already, i.e., for any $\varphi$, $\struct$, and $\sigma \in \events$
  we have that $(\struct, \sigma) \in \bad(\ltlX\varphi)$ iff
  $\calL(\varphi)_\struct = \emptyset$.
  The $\Leftarrow$-direction is obvious.  For the other direction:
  \[
  \begin{array}{ll}
    & (\struct, \sigma) \in \bad(\ltlX\varphi) \\
    \Rightarrow & \hbox{for all }\tstruct \sim \struct \hbox{ and } w \in \events^\omega, \hbox{ we have that } (\struct\tstruct, \sigma w) \not\models \ltlX\varphi \\
    \Rightarrow & \hbox{for all }\tstruct \sim \struct \hbox{ and } w \in \events^\omega, \hbox{ we have that } (\tstruct, w) \not\models \varphi \\
    \Rightarrow & \cL(\varphi)_\struct = \emptyset
    \hbox{ (which is generally undecidable by Lemma~\ref{lem:fo:restricted})}.
  \end{array}
  \]
  \qed
\end{proof}


\noindent\textbf{Lemma~\ref{lem:sa:correct}.  }
Let $\varphi \in \LTLFO$ (not necessarily a sentence) and $v$ be a
valuation.
For each accepting run $\rho$ in $\cAphiv$ over input
$(\tstruct, w)$, $\psi \in \cl(\varphi)$, and $i \geq 0$, we have that
$\psi \in \rho(i)$ iff $(\tstruct, w, v, i) \models \psi$.
\begin{proof}
  We proceed by a nested induction on $\depth(\varphi)$ and the
  structure of $\psi \in \cl(\varphi)$.
  For the base case let $\depth(\varphi) = 0$:
  We fix $\rho$ to be an accepting run in
  $\cAphiv$ over $(\tstruct, w)$, and proceed by induction over those
  formulae $\psi \in \cl(\varphi)$ which are of depth zero (i.e.,
  without quantifiers) since $\depth(\varphi) = 0$.  Therefore, this
  case basically resembles the correctness argument of B\"uchi
  automata for propositional LTL (cf.\
  \cite[\S5]{Baier:2008:PMC:1373322}).  For an arbitrary $i \geq 0$,
  we have
  \begin{itemize}
  \item $\psi = r(\vt)$:
    \[
    \begin{array}{lcl}
      r(\vt)
      \in \rho(i) & \Leftrightarrow & \vt^{I_i} \in r^{I_i} (\hbox{by the definition of } \ddelta),\\
      & & \hbox{where, as before, for any variable } x \hbox{ in } \vt, \hbox{ by } x^{I_i} \hbox{ we mean } v(x) \\
      & \Leftrightarrow &(\tstruct, w, v, i) \models r(\vt)\ (\hbox{by the semantics of } \LTLFO)
    \end{array}
    \]
  \item $\psi = p(\vt)$: analogous to the above.
  \item $\psi = \neg\psi'$:
    \[
    \begin{array}{ll}
      \neg\psi' \in \rho(i) & \Leftrightarrow \psi' \not\in \rho(i)\ (\hbox{by the completeness assumption of all } q \in Q)\\
      & \Leftrightarrow (\tstruct, w, v, i) \not\models \psi'\ (\hbox{by induction hypothesis})\\
      & \Leftrightarrow (\tstruct, w, v, i) \models \neg\psi'\ (\hbox{by the semantics of } \LTLFO)
    \end{array}
    \]
  \item $\psi = \psi_1 \wedge \psi_2$:
    \[
    \begin{array}{ll}
      \psi_1 \wedge \psi_2 \in \rho(i) & \Leftrightarrow \{\psi_1, \psi_2\} \subseteq \rho(i)\ (\hbox{by the completeness assumption of all } q \in Q)\\
      & \Leftrightarrow (\tstruct, w, v, i) \models \psi'_1 \hbox{ and } (\tstruct, w, v, i) \models \psi_2\ (\hbox{by induction hypothesis})\\
      & \Leftrightarrow (\tstruct, w, v, i) \models \psi_1 \wedge \psi_2\ (\hbox{by the semantics of } \LTLFO)
    \end{array}
    \]
  \item $\psi = \ltlX \psi'$:
    \[
    \begin{array}{ll}
      \ltlX \psi' \in \rho(i) & \Leftrightarrow \psi' \in \rho(i+1)\ (\hbox{by the definition of } \ddelta)\\
      & \Leftrightarrow (\tstruct, w, v, i+1) \models \psi'\ (\hbox{by induction hypothesis})\\
      & \Leftrightarrow (\tstruct, w, v, i) \models \ltlX\psi'\ (\hbox{by the semantics of } \LTLFO)
    \end{array}
    \]
  \item $\psi = \psi_1 \ltlU \psi_2$: we first show the
    $\Rightarrow$-direction.  For this, let us first show that there
    is a $j \geq i$, such that $(\tstruct, w, v, j) \models \psi_2$
    holds.  For suppose not, then for all $j \geq i$, we have that
    $(\tstruct, w, v, j) \not\models \psi_2$ and, consequently, by
    induction hypothesis $\psi_2 \not\in \rho(j)$.  By definition of
    $\ddelta$, since $\psi_1 \ltlU \psi_2 \in \rho(i)$ and there isn't
    a $j$ \st $\psi_2 \in \rho(j)$, we have that $\psi_1 \ltlU \psi_2
    \in \rho(j)$ for all $j \geq 0$.  On the other hand, $\rho$ is
    accepting in $\cAphi$, thus there exist infinitely many $j \geq
    i$, \st $\psi_1 \ltlU \psi_2 \not\in \rho(j)$ or $\psi_2 \in
    \rho(j)$ by the definition of the generalised B\"uchi acceptance
    condition $\calF$, which is a contradiction.  Let us, in what
    follows, fix the smallest such $j$.  We still need to show that
    for all $i \leq k \leq j$, $(\tstruct, w, v, k) \models \psi_1$
    holds.  As $j$ is the smallest such $j$, where $\psi_2 \in
    \rho(j)$ it follows that $\psi_2 \not\in \rho(k)$ for any such
    $k$.  As $\psi_1 \ltlU \psi_2 \in \rho(i)$, it follows by
    definition of $\ddelta$ that $\psi_1 \in \rho(i)$ and $\psi_1
    \ltlU \psi_2 \in \rho(i+1)$.  We can then inductively apply this
    argument to all $i \leq k < j$, such that $\psi_1 \in \rho(k)$ and
    $\psi_1 \ltlU \psi_2 \in \rho(k+1)$ hold.
    The statement then follows from the induction hypothesis.\\

    Let us now focus on the $\Leftarrow$-direction, i.e., suppose
    $(\tstruct, w, v, i) \models \psi_1 \ltlU \psi_2$ implies that $\psi_1
    \ltlU \psi_2 \in \rho(i)$.  By assumption, there is a $j \geq i$,
    such that $(\tstruct, w, v, j) \models \psi_2$ and for all $i \leq k <
    j$, we have that $(\tstruct, w, v, k) \models \psi_1$.  Therefore, by
    induction hypothesis, $\psi_2 \in \rho(j)$ and $\psi_1 \in \rho(k)$
    for all such $k$.  Then, by the completeness assumption of all $q
    \in Q$, we also get $\psi_1 \ltlU \psi_2 \in p_j$, and if $j = i$,
    we are done.  Otherwise with an inductive argument similar to the
    previous case on $k = j - 1$, $k = j - 2$, \ldots, $k = i$, we can
    infer that $\psi_1 \ltlU \psi_2 \in \rho(k)$.
  \end{itemize}
  Let $\depth(\varphi) = n > 0$, i.e., we suppose that our claim holds
  for all formulae with quantifier depth less than $n$.  We continue our
  proof by structural induction, where the quantifier free cases are
  almost exactly as above.  Therefore, we focus only on the following
  case.
  \begin{itemize}
  \item $\psi = \forall \vec{x}:p.\ \psi'$:
    for this case, as before with the $\ltlU$-operator, we will first
    show the $\Rightarrow$-direction, i.e., for all $i \geq 0$ we have
    $\forall \vec{x}:p.\ \psi' \in \rho(i)$ implies $(\tstruct, w, v,
    i) \models \forall \vec{x}:p.\ \psi'$. 
    By the semantics of \LTLFO, the latter is equivalent to for all
    $(p,\vd) \in w_i$, $(\tstruct, w, v \cup \{\vx \mapsto \vd\}, i)
    \models \psi'$. If there is no $(p, \vd) \in w_i$ the statement is
    vacuously true.
    Otherwise,
    there are some actions $(p,\vd) \in w_i$ and
    \[
    \spawn(\rho(i), (\struct_i, w_i)) = B \land \bigwedge_{(p, \vd) \in
      w_i}\cA_{\psi', v \cup \{\vx \mapsto \vd\}},
    \] 
    where $B$ is a Boolean combination of SAs corresponding to the
    remaining elements in $\rho(i)$.
    %
    %
    As $\rho$ is accepting in $\cAphiv$, there exists a $Y_i$
    satisfying $\spawn(\rho(i), (\struct_i, w_i))$, \st all $\cA \in
    Y_i$ have an accepting run on input $(\tstruct^i, w^i)$. It
    follows that $Y_i$ contains an automaton $\cA_{\psi', v \cup \{\vx
      \mapsto \vd\}}$ for each action $(p, \vd) \in
    w_i$ that has an accepting run $\rho'$.
    As the respective levels of these automata is $n-1$, we can
    use the induction hypothesis and note that the following holds
    true for each of the $\cA_{\psi', v \cup \{\vx \mapsto \vd\}} \in
    Y_i$:
    \[
    \hbox{for all: } \nu \in \cl(\psi') \hbox{ and } l \geq 0, \nu \in
    \rho'(l) \hbox{ iff } (\tstruct, w, v \cup \{\vx \mapsto \vd \},
    i+l) \models \nu,
    \]
    We can now set $\nu = \psi'$, respectively, and $l = 0$, from
    which it follows that $\psi' \in \rho'(0)$ iff $(\tstruct, w,v
    \cup \{\vx \mapsto \vd \}, i) \models \psi'$, respectively.  As by
    construction of an SA the initial states of runs contain the
    formula which the SA represents, we have $\psi' \in \rho'(0)$ and
    hence $(\tstruct, w, v \cup \{\vx \mapsto \vd\}, i) \models
    \psi'$, respectively. As this holds for all $\cA_{\psi', v \cup
      \{\vx \mapsto \vd\}}$, where $(p,\vd) \in w_i$, it follows by
    semantics of $\LTLFO$ that $(\tstruct, w, v, i) \models \forall
    \vec{x}:p.\ \psi'$.

    Let us now consider the $\Leftarrow$-direction, i.e., $(\tstruct,
    w, v, i) \models \forall \vec{x}:p.\ \psi'$ implies $\forall
    \vec{x}:p.\ \psi' \in \rho(i)$, which we show by
    contradiction. Suppose $\forall \vec{x}:p.\ \psi' \not\in
    \rho(i)$, which implies by the completeness assumption of all $q
    \in Q$ that $\neg\forall \vec{x}:p.\ \psi' \in \rho(i)$ holds.  If
    there is no $(p,\vd) \in w_i$, then $\spawn(\rho(i), (\struct_i,
    w_i))$ is equivalent to $\bot$ and $\rho$ could not be accepting.
    Therefore there must be some $(p,\vd) \in w_i$, \st
    \[
    \spawn(\rho(i), (\struct_i, w_i)) = B \land \bigvee_{(p,\vd) \in w_i}
    \cA_{\neg\psi', v \cup \{\vx \mapsto \vd\}},
    \]
    where $B$ is a Boolean combination of SAs
    corresponding to the remaining elements in $\rho(i)$.
    Because $\rho$ is accepting in $\cAphiv$, there exists a $Y_i$,
    such that $Y_i \models \spawn(\rho(i), (\struct_i, w_i))$, and there
    is at least one SA, $\cA'=\cA_{\neg\psi', v \cup \{\vx \mapsto
      \vd\}} \in Y_i$, with corresponding $(p, \vd) \in w_i$,
    \st $(\tstruct^i, w^i)$ is accepted by $\cA'$ as input; that is,
    $\cA'$ has an accepting run, $\rho'$, on said input.  As this
    automaton's level is $n - 1$, we can apply the induction
    hypothesis and obtain
    \[
    \hbox{for all: } \nu \in \cl(\neg\psi') \hbox{ and } l
    \geq 0, \nu \in \rho'(l) \hbox{ iff } (\tstruct, w, v \cup \{\vx \mapsto \vd\}, i+l)
    \models \nu.
    \]
    We can now set $\nu = \neg\psi'$ and $l = 0$, and
    since $\nu$ belongs to the initial states in accepting runs, we
    derive $(\tstruct, w, v \cup \{\vx \mapsto \vd\}, i) \models \neg\psi'$, which is
    a contradiction to our initial hypothesis.
    \qed
  \end{itemize}
\end{proof}

\noindent\textbf{Theorem~\ref{thm:sa:correct}.  }
The constructed SA is correct in the sense that for any sentence
$\varphi \in \LTLFO$, we have that $\calL(\cA_\varphi) =
\calL(\varphi)$.
\begin{proof}
  $\subseteq$:
  Follows from Lemma~\ref{lem:sa:correct}: let $\rho$ be an accepting
  run over $(\tstruct, w)$ in $\cAphi$.  By definition of an (accepting)
  run, $\varphi \in \rho(0)$, and therefore $(\tstruct, w) \in
  \cL(\varphi)$.

  $\supseteq$:
  We show the more general statement: Given a (possibly not closed)
  formula $\varphi \in \LTLFO$ and valuation $v$. It holds that
  $\{(\tstruct, w) \mid (\tstruct, w, v, 0) \models \varphi \}
  \subseteq \cL(\cAphiv)$.
  We define for all $i \geq 0$ the set $\rho(i) = \{ \psi \in
  \cl(\varphi) \mid (\tstruct, w, v, i) \models \psi \}$ for some
  arbitrary but fixed formula $\varphi \in \LTLFO$ and valuation $v$,
  and arbitrary but fixed $(\tstruct, w)$, where $(\tstruct, w, v, 0)
  \models \varphi$. Let us now show that $\rho = \rho(0) \rho(1)
  \ldots$ is a well-defined run in $\cAphiv$ over $(\tstruct, w)$:
  Firstly, from the construction of $Q$, it follows that for all $i$,
  $\rho(i) \in Q$. Secondly, since $\varphi \in \cl(\varphi)$ and
  $(\tstruct, w, v, 0) \models \varphi$, $\rho(0)$ always contains
  $\varphi$. Thirdly, $\rho(i+1) \in \ddelta(\rho(i), (\struct_i,
  w_i))$ holds for all $i$.  The latter is the case iff
  \begin{itemize}
  \item for all $\ltlX\psi \in \cl(\varphi)$: $\ltlX\psi \in \rho(i)$ iff
    $\psi \in \rho(i+1)$, and
  \item for all $\psi_1 \ltlU \psi_2 \in \cl(\varphi)$: $\psi_1 \ltlU
    \psi_2 \in \rho(i)$ iff $\psi_2 \in \rho(i)$ or ($\psi_1 \in \rho(1)$
    and $\psi_1 \ltlU \psi_2 \in \rho(i + 1)$).
  \end{itemize}
  The first condition can be shown as follows:
  \[
  \begin{array}{ll}
    \ltlX \psi \in \rho(i) & \Leftrightarrow (\tstruct, w, v, i) \models \ltlX\psi\ (\hbox{by definition of } \rho(i))\\
    & \Leftrightarrow (\tstruct, w, v, i+1) \models \psi\ (\hbox{by the semantics of \LTLFO})\\
    & \Leftrightarrow \psi \in \rho(i+1)\ (\hbox{by the definition of } \rho(i+1)).
  \end{array}
  \]
  The second can be shown as follows:
  \[
  \begin{array}{ll}
    \psi_1 \ltlU \psi_2 \in \rho(i) & \Leftrightarrow (\tstruct, w, v, i) \models \psi_1\ltlU\psi_2\ (\hbox{by definition of } \rho(i))\\
    & \Leftrightarrow (\tstruct, w, v, i) \models \psi_2 \vee (\psi_1 \wedge \ltlX(\psi_1 \ltlU\psi_2)) \\
    & \Leftrightarrow (\tstruct, w, v, i) \models \psi_2 \hbox{ or } ((\tstruct, w, v, i) \models \psi_1 \hbox{ and } (\tstruct, w, v, i+1) \models \psi_1 \ltlU\psi_2)\\
    & \Leftrightarrow \psi_2 \in \rho(i) \hbox{ or } (\psi_1 \in \rho(1) \hbox{ and } \psi_1 \ltlU \psi_2 \in \rho(i+1))\ (\hbox{by definition of } \rho).
  \end{array}
  \]
  It remains to show that $\rho$ is also accepting in $\cAphiv$.
  We proceed by induction on $\depth(\varphi)$.  In what follows, let
  $\depth(\varphi) = 0$, i.e., we are showing local acceptance only.
  By the definition of acceptance we must have that for all
  $\psi_1\ltlU\psi_2 \in \cl(\varphi)$, there exist infinitely many $i
  \geq 0$, \st $\rho(i) \in F_{\psi_1\ltlU\psi_2}$, where
  $F_{\psi_1\ltlU\psi_2} \in \calF$.  For suppose not, i.e., there are
  only finitely many such $i$, then there is a $k \geq 0$, \st for all
  $j \geq k$ we have $\rho(j) \not\in F_{\psi_1\ltlU\psi_2}$ and
  therefore $\psi_1\ltlU\psi_2 \in \rho(j)$ and $\psi_2 \not\in\rho(j)$
  by definition of $F_{\psi_1\ltlU\psi_2}$.  In particular, from
  $\psi_1\ltlU\psi_2 \in \rho(k)$ we derive by construction of $\rho(k)$
  that there must be some $g \geq k$, \st $(\tstruct^g, w^g) \in
  \cL(\psi_2)$ and thus $\psi_2 \in \rho(k)$ with $g \geq k$.
  Contradiction.\\

  Let us now assume the statement holds for all formulae with depth
  strictly less than $n$ and assume $\depth(\varphi) = n$, where $n >
  0$.  We don't show local acceptance of $\rho$ as it is virtually the
  same as in the base case, and instead go on to show that for all $i
  \geq 0$, there is a $Y_i$, \st $Y_i \models \spawn(\rho(i),
  (\struct_i, w_i))$ and all $\cA \in Y_i$ are accepting $(\tstruct^i,
  w^i)$.
  Let us define the following two sets:
  \[
  Y_i^\forall = \{ \cA_{\psi, v \cup \{\vx \mapsto \vd\}} \mid \forall \vec{x}:p.\ \psi \in \rho(i) \hbox{ and } (p,\vd) \in w_i \}
  \]
  and
  \[
  \begin{array}{lll}
    Y_i^\exists  = \{ \cA_{\neg\psi, v \cup \{\vx \mapsto \vd\}} & \mid &
    \neg\forall \vec{x}:p.\ \psi \in \rho(i), (p,\vd)
    \in w_i, \\
    & & \hbox{ and } (\tstruct, w, v \cup \{\vx \mapsto \vd\}, i)
    \not\models \psi\}.
  \end{array}
  \]
  Set $Y_i = Y_i^\forall \cup Y_i^\exists$, which by construction
  satisfies $\spawn(\rho(i), (\struct_i, w_i))$.  We still need to
  show that every automaton in this set accepts $(\tstruct^i, w^i)$.
  Now for $\cA_{\nu,v \cup \{\vx \mapsto \vd\}} \in Y_i$ we have
  either $\nu =\psi$ for some $\forall \vec{x}:p.\ \psi \in \rho(i)$
  and $(p, \vd) \in w_i$, or $\nu = \neg\psi$ for some $\neg\forall
  \vec{x}:p.\ \psi \in \rho(i)$ and $(p, \vd) \in w_i$ \st $(\tstruct,
  w, v \cup \{\vx \mapsto \vd\}, i) \not\models \psi$ holds.  In
  either case by definition of $\rho(i)$ and semantics of $\LTLFO$, it
  follows that $(\tstruct, w, v \cup \{\vx \mapsto \vd\}, i) \models
  \nu$.  Since the level of $\cA_{\nu,v \cup \{\vx \mapsto \vd\}}$ is
  strictly less than $n$, we can apply the induction hypothesis and
  construct an accepting run for $(\tstruct^i, w^i)$, where
  $(\tstruct, w, v \cup \{\vx \mapsto \vd\}, i) \models \nu$, in
  $\cA_{\nu,v \cup \{\vx \mapsto \vd\}}$. The statement follows.
  \qed
\end{proof}

\noindent\textbf{Theorem~\ref{thm:mon:correct}.  }
  $M_{\varphi}(\tstruct, u) = \top \Rightarrow$ $(\tstruct, u) \in
  \good(\varphi)$ (resp.\ for $\bot$ and $\bad(\varphi)$).
\begin{proof}
  We prove the more general statement $M_{\varphi,v}(\tstruct, u) =
  \top \Rightarrow$ $(\tstruct, u) \in \good(\varphi,v)$, where
  $\varphi$ possibly has some free variables and $v$ is a valuation,
  by a nested induction over $\depth(\varphi)$.
  \begin{itemize}
  \item For the base case let $\depth(\varphi)=0$, where $\varphi$
    possibly has free variables, $(\tstruct, u)$ be an arbitrary but
    fixed prefix and $v$ a valuation. Suppose
    $M_{\varphi,v}(\tstruct,u)$ returns $\top$ after processing
    $(\tstruct, u)$, but $(\tstruct, u) \not\in \good(\varphi,v)$.  By
    \ref{algo:m:verdict} and \ref{itm:result}, the buffer of
    $T_{\lnot\varphi,v}$ is empty, i.e., $B_{\lnot\varphi,v}=
    \emptyset$.  By \ref{algo:r:buffer} and because
    $\cA_{\neg\varphi,v}$ has an accepting run $\rho$ over $(\tstruct,
    u)$ with some suffix, $B_{\neg\varphi,v}$ contains $(\rho(|u|),
    [\top])$ after processing $(\tstruct, u)$. Furthermore, because
    $\spawn$ yields $\top$ for any input iff $\depth(\lnot\varphi) =
    0$, no run in the buffer is ever removed in
    \ref{itm:rm_seq}. Contradiction.
  \item Let $depth(\varphi) > 0$, $(\tstruct, u)$ be an arbitrary but
    fixed prefix and $v$ a valuation.
    Under the same assumptions as above, we will reach a contradiction
    showing that after processing $(\tstruct, u)$, there is a sequence
    of obligations $(\rho(|u|), \allowbreak [obl_0, \ldots,
    \allowbreak obl_n])$ in buffer $B_{\lnot\varphi,v}$, which
    corresponds to an accepting run $\rho$ in $\cA_{\neg\varphi,v}$
    over $(\tstruct, u)$ with some suffix $(\tstruct',w')$. That is,
    M$_{\varphi,v}$ cannot return $\top$, after $B_{\lnot\varphi,v}$
    is empty, and $B_{\lnot\varphi,v}$ containing the above mentioned
    sequence at the same time.
    By \ref{algo:r:buffer}, $B_{\neg\varphi,v}$ contains a sequence
    $(\rho(|u|), \allowbreak [obl_0, \ldots, \allowbreak obl_n])$ that
    was incrementally created processing $(\tstruct, u)$ wrt.\
    $\ddelta$, eventually with some obligations removed if they were
    detected to be met by the input. 
    We now show that this sequence is never removed from the buffer in
    \ref{itm:rm_seq}.
    Suppose the run has been removed, then there was an
    $obl_j=\spawn(\rho(j), (\tstruct_j, u_j))$, that is
    \[
    \left( \bigwedge_{\forall \vx:p. \psi \in \rho(j)} \left(
        \bigwedge_{(p,\vd) \in u_j} \cA_{\psi,v'} \right) \right)
    \wedge \left( \bigwedge_{\neg\forall \vx:p. \psi \in \rho(j)}
      \left( \bigvee_{(p,\vd) \in u_j} \cA_{\neg\psi,v''} \right)
    \right),
    \]
    with $v'=v \cup \{\vx \mapsto \vd\}$ and $v''=v \cup \{\vx \mapsto
    \vd\}$, evaluated to $\bot$ after $l$ steps, with $0 \leq j \leq l
    < |u|$.
    That is, at least one submonitor corresponding to an automaton in
    the second conjunction has returned $\bot$ (or all submonitors
    corresponding to automata in a disjunction, for which the
    following argument would be similar).
    Wlog. let $\forall \vec{x}:p. \psi \in \rho(j)$, $(p,\vd) \in
    u_j$, and M$_{\psi,v'}(\struct_j, \ldots, \allowbreak \struct_l,
    u_j, \ldots, u_l)=\bot$, i.e., M$_{\psi,v'}$ is the submonitor
    corresponding to $\cA_{\psi,v'}$.
    As $level(\psi) < level(\varphi)$, from the induction hypothesis
    follows that $(\struct_j, \ldots, \struct_l, \allowbreak u_j,
    \ldots, u_l) \in bad(\psi,v')$, i.e., $(\struct_j, \ldots,
    \struct_l\tstruct'', u_j, \ldots, u_lw'') \models \psi$ with
    evaluation $v'$ for any $(\tstruct'', w'')$, and therefore
    $(\struct_j, \ldots, \struct_l\tstruct'', u_j, \ldots, u_lw'')
    \models \neg\forall x:p.\psi$ under valuation $v$.
    But as $\rho$ over $(\tstruct\tstruct', uw')$ is an accepting run
    in $\cA_{\neg\varphi,v}$ and $\forall x:p.\psi \in \rho(j)$, it
    follows that $(\tstruct^j\tstruct', u^jw') \models \forall
    x:p.\psi$. Now, we choose $(\tstruct'', w'')$ to be
    $(\tstruct_{l+1},\ldots,\tstruct_{|u|}\tstruct', \allowbreak
    u_{l+1},\ldots,u_{|u|}w')$. Contradiction.
  
    As for our second statement above, it can be shown similar as
    before. \qed
  \end{itemize}
\end{proof}



\end{document}